\documentclass[10pt,a4paper]{article}
\usepackage{quantikz}
\usepackage{float}
\usepackage{subcaption}
\usepackage{algorithm}
\usepackage{algpseudocode}
\usepackage{tikz}
\usetikzlibrary{decorations.pathmorphing, arrows.meta, positioning, backgrounds, fit}
\newcommand{\hdfont}{}
\tikzset{
  hdsketch/.style={},
  hdblock/.style={hdsketch, draw=black, line width=0.7pt, fill=#1,
                  text width=2.35cm, minimum height=1.2cm, align=center,
                  inner sep=2pt, font=\hdfont\small},
  hdadd/.style={hdsketch, draw=black, line width=0.7pt, dash pattern=on 3pt off 2pt,
                fill=#1, text width=2.35cm, minimum height=1.2cm, align=center,
                inner sep=2pt, font=\hdfont\small},
  hdarr/.style={-{Latex[length=2.2mm,width=1.7mm]}, hdsketch, line width=0.7pt},
  hdio/.style={font=\hdfont\small, align=center, text width=2.0cm},
  hdwire/.style={hdsketch, line width=0.6pt},
  hdgate/.style={hdsketch, draw=black, line width=0.7pt, fill=#1,
                 align=center, inner sep=3pt, font=\hdfont\scriptsize, rounded corners=1pt},
  hdlane/.style={font=\hdfont\scriptsize, align=right, text width=1.5cm, anchor=east},
  hdreg/.style={font=\hdfont\footnotesize\itshape}
}

\newcommand{\qgatfrontend}[1]{%
  \node[hdio] (#1ain) at (-0.25,1.4) {jet\\ constituents};
  \node[hdblock={gray!14}, text width=2.0cm, minimum height=1.3cm] (#1enc) at (2.9,1.4)
       {Classical encoder\\ $\mv{x}_i\!\to\!(\theta_i,\eta_g)$};
  
  \draw[hdarr] (#1ain.east) -- (#1enc.west);
  \node[hdlane] at (3.85,0) {$q_0\!\dots\! q_{N-1}$ (data)};
  \node[hdlane, text width=1.9cm] at (3.85,-1.2) {$q_T$ (3-particle\\ product)};
  \draw[hdwire] (4.5,0) -- (13.9,0);
  \draw[hdwire] (4.5,-1.2) -- (13.9,-1.2);
  \node[hdgate={blue!14}, text width=2.6cm, minimum height=1.7cm] (#1att) at (6.4,0)
       {Grade-aware\\ attention\\ $\mathrm{IsingZZ}(\eta_g)$;\\ $R_YR_Z$};
  \draw[hdwire] (#1enc.east) -- (5.2,1.4);
  \fill (5.2,1.4) circle (1.5pt);
  \draw[hdarr] (5.2,1.4) -- node[pos=0.5,left=3pt,font=\hdfont\scriptsize]{$\theta_i$} ([xshift=-18pt]#1att.north);
  \draw[hdarr] (5.2,1.4) -- node[pos=0.5,right=3pt,yshift=7pt,font=\hdfont\scriptsize]{$\eta_g$} ([xshift=18pt]#1att.north);
  \node[hdgate={cyan!16}, minimum width=0.55cm, minimum height=0.55cm] (#1tH) at (8.1,-1.2) {$H$};
  \node[hdgate={cyan!12}, text width=1.5cm, minimum height=0.95cm] (#1tM) at (8.9,0) {ctrl-MultiRZ\\ $f_{ijk}$};
  \fill (8.9,-1.2) circle (1.7pt);
  \draw[hdwire] (8.9,-1.2) -- (8.9,-0.475);
  \node[hdgate={green!16}, text width=1.6cm, minimum height=1.9cm] (#1vqc) at (12.9,-0.6)
       {VQC\\ $R_YR_Z$ + CNOT ring\\ (data ${+}$ $q_T$)};
}
\usepackage[utf8]{inputenc} 
\usepackage[T1]{fontenc} 	
\usepackage[english]{babel} 

\usepackage[bitstream-charter]{mathdesign}

\usepackage{geometry} 		
\usepackage{amsmath} 		
\usepackage{mathtools} 		
\usepackage{float} 			
\usepackage{graphicx} 		
\usepackage{tabularx} 		
\usepackage{booktabs} 		
\usepackage{color, xcolor} 	
\usepackage{pdfpages} 		
\usepackage{extarrows} 		
\usepackage{multirow} 		
\usepackage{multicol} 		
\usepackage{enumitem} 		
\usepackage{xspace} 		
\usepackage{stackrel} 		
\usepackage{tikz} 			
\usetikzlibrary{positioning} 
\usepackage{braket} 		
\usepackage{bm} 			
\usepackage{tensor} 		
\usepackage{slashed} 		
\usepackage{siunitx} 		
\usepackage{lastpage} 		
\usepackage{cite} 			
\usepackage[normalem]{ulem} 
\usepackage{fontawesome} 	
\usepackage{tocloft} 		
\usepackage{titlesec} 		
\usepackage{doi} 			
\usepackage{hyperref} 		
\usepackage[most]{tcolorbox} 					
\usepackage[nameinlink, capitalize]{cleveref} 	
\usepackage[nottoc, notlot, notlof]{tocbibind} 	
\usepackage[ruled, vlined, algo2e]{algorithm2e} 
\usepackage{makecell}
\usepackage[makeroom]{cancel}
\usepackage{feynmf}
\usepackage{booktabs}
\usepackage{pdfpages}
\usepackage{pifont}
\makeatletter
\def\BState{\State\hskip-\ALG@thistlm}
\makeatother

\makeatletter
\@ifundefined{pdfoutput}{}{\DeclareGraphicsRule{*}{mps}{*}{}}
\makeatother

\makeatletter
\DeclareRobustCommand*{\bfseries}{%
   \not@math@alphabet\bfseries\mathbf
   \fontseries\bfdefault\selectfont
   \boldmath
}
\makeatother

\hypersetup{
	pdftitle={MadNIS},
	pdfauthor={Heimel et al.},
	colorlinks=true, 			
	linkcolor={red!50!black}, 	
	citecolor={blue!50!black}, 	
	urlcolor={blue!80!black} 	
}

\DeclareSymbolFont{usualmathcal}{OMS}{cmsy}{m}{n}
\DeclareSymbolFontAlphabet{\mathcal}{usualmathcal}

\SetArgSty{textnormal}
\SetKwComment{Comment}{{\small\#}~}{}
\SetCommentSty{mycommfont}

\setitemize{itemsep=0pt, parsep=0pt} 				
\setenumerate{itemsep=0pt, parsep=0pt} 				
\setitemize{itemsep=2pt,topsep=2pt,parsep=0pt,partopsep=0pt,leftmargin=*}
\setenumerate{itemsep=0pt,topsep=2pt,parsep=0pt,partopsep=0pt,labelindent=3pt,leftmargin=*}
\usepackage{amsmath}
 
\usepackage{amsthm} 
\usepackage{amsmath,amsthm}
\usepackage{graphicx}
\usepackage{booktabs}
\usepackage{multirow}
\usepackage{hyperref}
\usepackage{xcolor}
\usepackage{algorithm}
\usepackage{algpseudocode}
\usepackage{tikz}
\usepackage{pgfplots}
\pgfplotsset{compat=1.17} 
\usetikzlibrary{shapes,arrows,positioning,calc,decorations.pathreplacing,fit,backgrounds,shapes.geometric,shapes.symbols,shadows,arrows.meta}
\usepackage{subcaption}
\usepackage{rotating}
\usepackage{physics}
\usepackage{bbm}
\usepackage[utf8]{inputenc}
\usepackage[T1]{fontenc}
\usepackage{microtype}
\theoremstyle{definition}

\newtheorem{lemma}{Lemma}

\usepackage{quantikz}

\definecolor{red_cb}{HTML}{e41a1c}
\definecolor{blue_cb}{HTML}{377eb8}
\definecolor{green_cb}{HTML}{4daf4a}
\definecolor{purple_cb}{HTML}{984ea3}
\definecolor{orange_cb}{HTML}{ff7f00}

\definecolor{EmeraldGreen}{HTML}{1ea78d}
\definecolor{EnglishRed}{HTML}{b02427}
\hypersetup{colorlinks=true,urlcolor=EmeraldGreen,citecolor=EmeraldGreen,linkcolor=EnglishRed}

\newcommand\one{\leavevmode\hbox{\small1\normalsize\kern-.33em1}}

\newcommand{\arXiv}[2][]{%
	\ifthenelse{\equal{#1}{}}%
	{\href{http://arxiv.org/abs/#2}{arXiv:#2}}%
	{\href{http://arxiv.org/abs/#2}{arXiv:#2~[#1]}}}

\def\slashchar#1{\setbox0=\hbox{$#1$}           
   \dimen0=\wd0                                 
   \setbox1=\hbox{/} \dimen1=\wd1               
   \ifdim\dimen0>\dimen1                        
      \rlap{\hbox to \dimen0{\hfil/\hfil}}      
      #1                                        
   \else                                        
      \rlap{\hbox to \dimen1{\hfil$#1$\hfil}}   
      /                                         
   \fi}

\newcommand{\tikznode}[2]{%
\ifmmode%
\tikz[remember picture,baseline=(#1.base),inner sep=0pt] \node (#1) {$#2$};%
\else
\tikz[remember picture,baseline=(#1.base),inner sep=0pt] \node (#1) {#2};%
\fi}

\def\mathswitchr#1{\relax\ifmmode{\mathrm{#1}}\else$\mathrm{#1}$\xspace\fi}
\def\mathswitch#1{\relax\ifmmode#1\else$#1$\xspace\fi}

\hypersetup{pdftitle={QGrAT: Quantum Grade-Aware Transformer for Jet Tagging},
            pdfauthor={Pallavi Bhardwaj, Michael Spannowsky}}
\graphicspath{{./figs/}}
\usepackage{fancyhdr}
\fancypagestyle{titlefooter}{%
  \fancyhf{}%
  \fancyfoot[L]{%
    \scriptsize
    \begin{tabular}[t]{@{}l@{}}
      \hline\noalign{\vskip 2pt}
      \textsuperscript{*}\ \texttt{pallavi.bhardwaj@stud.uni-heidelberg.de}
    \end{tabular}}%
}
\newcommand{\Cl}{\mathcal{C}\ell}

\newcommand{\mv}[1]{\mathbf{#1}}
\newcommand{\grade}[2]{\langle#1\rangle_{#2}}
\newcommand{\RR}{\mathbb{R}}
\newcommand{\CC}{\mathbb{C}}
\newcommand{\qgat}{\text{QGrAT}}        
 
\newcommand{\lgatr}{\text{L-GATr}}

\begin{document}
\thispagestyle{titlefooter}
\begin{center}{\Large \textbf{
QGrAT: Quantum Grade-Aware Transformer for Jet Tagging
}}
\end{center}

\begin{center}
  Pallavi Bhardwaj\textsuperscript{1,2,*} and 
  Michael Spannowsky\textsuperscript{3,4}
\end{center}

\begin{center}
{\bf 1} Institut f\"ur Theoretische Physik, Universit\"at Heidelberg, Germany
\\
{\bf 2} SAP SE, Walldorf, Germany
\\
{\bf 3} Institute for Theoretical Physics (ITP), Karlsruhe Institute of Technology (KIT),
76131~Karlsruhe, Germany
\\
{\bf 4} Institute for Quantum Materials and Technologies (IQMT), Karlsruhe Institute of Technology (KIT),
76131~Karlsruhe, Germany
\end{center}
\begin{center}
\end{center}

\section*{Abstract}
Lorentz symmetry and the three-prong structure of hadronic top decays provide physical
guidance for designing compact quantum jet taggers. We introduce QGrAT (Quantum
Grade-Aware Transformer), a quantum transformer that incorporates these structures into
its interactions between constituent qubits. Geometric features constructed from the
constituent four-momenta are organised by the grades of the spacetime Clifford algebra
$\Cl(1,3)$, and their grade-specific inner products determine the pairwise attention
couplings. For fixed constituents and co-transforming reference vectors, these couplings
are Lorentz invariant, although the component-wise input encoding makes the complete
classifier frame dependent. An auxiliary qubit introduces an explicit three-particle
interaction through products of pairwise Minkowski inner products, while a trainable
in-circuit readout converts the resulting quantum state into a classification score.
We evaluate the architecture in exact state-vector simulation on the Top Tagging
benchmarks. With ten constituents per jet, the configuration
using grades $0$--$3$ and the three-particle ancilla requires twelve qubits and $144$
trainable parameters, reaching AUC $0.944{\pm}0.002$ and background rejection
$1/\varepsilon_B=35.7{\pm}2.4$ ($96.8{\pm}7.7$) at signal efficiency
$\varepsilon_S=0.50$ ($0.30$). This performance is competitive with the classical
\lgatr{}-Slim baseline at matched constituent and parameter budgets. Adding the
three-particle ancilla raises background rejection at $50\%$ signal efficiency from
$19.9$ to $35.7$, whereas extending the representation to the tested grade-$4$ feature
gives no consistent improvement. The study establishes a framework for incorporating
spacetime geometry into quantum attention and identifies explicit three-particle
correlations as the principal source of improved background rejection among the
architectural variations tested.

\section{Introduction}
\label{sec:intro}
At the Large Hadron Collider (LHC), searches for new particles rely on inferring the
underlying interactions from the particles observed in the
detectors~\cite{Spannowsky:2025tvb}. Jets are central to this task: they are narrow sprays
of particles formed as energetic quarks and gluons radiate and turn into bound states
called hadrons. Grouping these particles into a jet gives access to their combined
four-momentum, while the distribution of energy and the relative directions of particles
within the jet carry information about its origin. The study of these internal patterns
is known as jet substructure~\cite{Marzani:2019hun,Kogler:2018hem}.
Jet substructure is particularly useful for highly boosted heavy particles, such as
top quarks, whose momenta are much larger than their masses. Their decay products then
tend to travel in similar directions and can be reconstructed within a single jet.
A top quark that decays into a bottom quark and two light quarks,
$t\to bW\to bq\bar q'$, can leave three concentrated regions of energy, or prongs.
Top tagging aims to distinguish these jets from
the far more abundant jets initiated by light quarks and gluons, whose radiation can
produce similar patterns~\cite{Marzani:2019hun}. Reliable tagging helps isolate processes
of interest and improves sensitivity to new heavy particles that decay into top
quarks~\cite{Spannowsky:2025tvb}.
These differences reside in correlations among jet constituents, motivating machine
learning models that process their four-momenta $p_i^\mu=(E_i,\vec p_i)$ directly,
where $E_i$ and $\vec p_i$ are the energy and three-momentum of constituent $i$.
These components change under Lorentz transformations, which describe rotations and
changes between inertial reference frames. Lorentz-invariant combinations, such as
pairwise invariant masses, retain the same value in every frame, whereas equivariant
representations transform in a prescribed way alongside the input momenta.
These properties provide a physical guide for designing how constituents interact
inside a learning model.

Classical Lorentz-equivariant models for jet tagging have evolved rapidly.
LorentzNet~\cite{lorentznet} constructs a message-passing graph network with Minkowski
inner products as invariant features, while PELICAN~\cite{pelican} uses permutation-equivariant
aggregation with Lorentz invariants.
The Particle Transformer~\cite{particle_transformer} applies multi-head self-attention with
pairwise interaction terms.
GATr~\cite{gatr} introduced geometric algebra transformers acting on multivectors in
$\Cl(3,0,1)$, and L-GATr~\cite{lgatr} extended this to the Minkowski algebra $\Cl(1,3)$,
embedding all operations as equivariant maps on Lorentz multivectors and reporting
state-of-the-art performance on the Top Tagging benchmark.

In parallel, quantum and quantum-inspired approaches have been developed to analyse collider-event data, which 
encompass supervised classification~\cite{blance_spannowsky_vqc,wu_quantum_hep,maier_photonic_qelm,felser_quantum_inspired_hep} and anomaly
detection~\cite{ngairangbam_quantum_autoencoder,bermot_quantum_gan_hep,blance_photonic_graphs,incudini_quantum_kernels,belis_latent_quantum_anomaly}.
Related developments extend to neural-network training on quantum
annealers~\cite{abel_completely_quantum_nn}. Due to the limited number of qubits available on
 current quantum hardware, a central design choice is how to represent collision data in a compact form. Strategies include classical dimensionality reduction, direct particle-level
encodings, and representations based on jet images or Lund-plane coordinates~\cite{belis_guided_quantum_compression,bal_one_particle_one_qubit,elhag_qcnn_jets,lp2b_encoding}.
The circuit must then balance expressive power and trainability against its number of gates and
parameters~\cite{ngairangbam_geometric_qnn,ngairangbam_generic_reachability}; dimensional expressivity analysis can help identify redundant parameters
and simplify the ansatz~\cite{funcke_dimensional_expressivity}.
Another limitation for current quantum computers is the coherence time, which affects the circuit depth, and suggests an efficient and lean model design. Graph structure, physical symmetries and attention mechanisms provide useful inductive
biases for quantum jet classifiers~\cite{forestano_invariant_equivariant_qgnn,suda_neto_lie_eqgnn,li_enforcing_symmetries,qcgnn_jet,comajoan_quantum_vision_jets}.
Geometric quantum machine learning supplies methods for building group invariance and
equivariance into quantum models~\cite{meyer_group_invariant,larocca_gqml,nguyen_equivariant}.
Such constraints can also improve trainability, including the avoidance of barren
plateaus in suitable $S_n$-equivariant circuit families~\cite{schatzki_sn}.
Quantum attention and Clifford-algebra representations offer complementary starting
points for combining relational learning with geometric
structure~\cite{quantum_attention_transformer,trindade_clifford_qnn}.

These classical and quantum developments motivate using the physical structure already
present in jet data to guide the design of a quantum tagger. The underlying interactions
respect Lorentz symmetry, and the constituent four-momenta transform in a known way under
changes of inertial frame. For a fixed set of constituents, these transformations preserve
invariant masses and Minkowski inner products, which encode information about the jet's
internal structure. Building such relations into a tagger supplies an inductive bias:
known physical structure is provided by the architecture, reducing what must be inferred
from training examples. This can support smaller models with fewer trainable parameters
and improve data efficiency and generalisation, as demonstrated by classical
Lorentz-aware taggers~\cite{lorentznet,Bogatskiy:2020pex}.

These potential gains are particularly valuable for quantum networks, where available
qubits and usable circuit depth are limited. Fixing part of the computation through
physical symmetry can reduce the number of independent parameters to optimise and focus
the circuit on correlations relevant to the task. The resulting expectation is that a
physically structured quantum tagger can achieve useful discrimination with a more
compact model and more efficient training. This motivates geometric quantum machine
learning, which constrains quantum architectures using the symmetries of the learning
problem~\cite{larocca_gqml,nguyen_equivariant,meyer_group_invariant,schatzki_sn}.
Clifford algebra provides a concrete link between the classical and quantum constructions:
it organises four-momenta and their correlations into multivectors for classical
equivariant networks~\cite{lgatr,Brehmer:2024LHC}, while its matrix representations
connect geometric information to qubit operators. Lorentz-scalar inner products can
then determine the angles of quantum gates.

We introduce QGrAT to investigate whether Lorentz-informed interactions can support competitive 
top tagging with a compact quantum model. The circuit couples pairs of jet constituents through gate angles 
calculated from Clifford inner products, evaluated separately for each multivector grade. 
These input-dependent weights act as attention scores by controlling the interactions between constituent 
qubits, without requiring a learned query–key map. Their symmetry properties depend on the constituent and 
reference vectors used to construct them. For a fixed constituent set, the grade-0, grade-1 and grade-2 
couplings are Lorentz invariant. The grade-3 feature is a momentum-weighted sum of trivectors. The grade-4 feature
uses four-volume magnitudes with a leading-constituent reference, and its invariance additionally
requires preserving that reference selection. 
The component-wise encoding retains a dependence on the reference frame, so the complete classifier is not 
Lorentz equivariant. We evaluate the model in exact simulation.

The paper is organised as follows. Section~\ref{sec:background} introduces spacetime
Clifford algebra, its multivector grades and its relation to qubit operators.
Section~\ref{sec:equivariance} establishes the symmetry properties of the attention
couplings and explains the role of reference vectors and constituent selection.
Section~\ref{sec:method} develops the circuit, from multivector encoding and grade-wise
attention to the ancilla couplings, variational layer and readout.
Section~\ref{sec:experimental_setup} describes the Top Tagging dataset, training procedure,
classical baselines and evaluation protocol. Section~\ref{sec:results} presents the tagging
performance and training diagnostics, while Section~\ref{sec:systematic_experiments}
examines the effects of grade content, ancilla couplings and constituent count through
controlled comparisons. Section~\ref{sec:conclusion} discusses the findings, limitations and
prospects for further development.

Our methodological contribution is, to our knowledge, the first quantum jet tagger
whose attention is organised by the grades of the spacetime Clifford algebra $\Cl(1,3)$,
with pairwise entangling-gate angles set by the corresponding Clifford inner products.
We combine this mechanism with a Hadamard-test ancilla carrying a Lorentz-invariant
three-particle correlation and an in-circuit quantum readout, and establish the symmetry
properties of the resulting attention couplings. In exact simulation, this construction
achieves top-tagging performance competitive with the classical \lgatr{}-Slim baseline
at matched constituent and parameter budgets. The study identifies the
three-particle-product ancilla as the main source of improved background rejection
among the architectural variations tested, demonstrating how a specific physical
correlation can strengthen a compact quantum classifier.

\section{Spacetime Algebra and Lorentz Symmetry}
\label{sec:background}

To build \qgat{}, we first need to describe jet kinematics in a form that makes its
Lorentz structure explicit. Our inputs are particle four-momenta, while the attention
mechanism also uses geometric quantities constructed from combinations of these momenta.
Clifford algebra lets us describe these objects in a common language and organise them
into multivector grades. In this section, we introduce the grades and their geometric
meaning, then explain how they transform under changes of inertial frame. This provides
the basis for identifying the invariant attention couplings and the conditions under
which their invariance holds in Section~\ref{sec:equivariance}. Finally, we recall how
Clifford algebra can be represented using Pauli operators, providing the algebraic
motivation for the quantum circuit constructed in Section~\ref{sec:method}.

\subsection{The Clifford Algebra of Minkowski Spacetime}
\label{sec:clifford}

The Clifford algebra $\Cl(1,3)$ over the Minkowski space $\RR^{1,3}$ is the associative
algebra generated by four basis vectors $\{e_0,e_1,e_2,e_3\}$ satisfying the
anticommutation relations~\cite{doran_lasenby}
\begin{equation}
  e_\mu e_\nu + e_\nu e_\mu = 2\eta_{\mu\nu}\,\mathbb{1},
  \qquad
  \eta = \mathrm{diag}(+1,-1,-1,-1).
  \label{eq:clifford_relation}
\end{equation}
The defining relation~\eqref{eq:clifford_relation} encodes the Minkowski metric directly
into the algebra, where $e_0^2=+1$ corresponds to the temporal direction and $e_i^2=-1$ to the
spatial directions.
The algebra has $2^4=16$ basis blades organised into five grade subspaces, with dimensions
$1+4+6+4+1$ corresponding to scalars, four-vectors, bivectors, trivectors, and the
pseudoscalar respectively~\cite{doran_lasenby}.
A general multivector $\mv{x}\in\Cl(1,3)$ is an arbitrary element of this 16-dimensional
space, and its grade-$k$ projection is written $\grade{\mv{x}}{k}$.
A particle four-momentum $p^\mu=(E,p_x,p_y,p_z)$ is embedded as the grade-1 element
$\mv{p}=p^\mu e_\mu$.

The geometric meaning of the grades is illustrated in Figure~\ref{fig:clifford_grades}.
An exterior product $a_1\wedge\cdots\wedge a_g$ forms a simple grade-$g$ blade spanning
the oriented $g$-dimensional subspace defined by its vector factors. It vanishes if the
factors are linearly dependent, and exchanging two factors reverses its sign.
Thus, adding an independent vector factor extends an oriented area to a three-volume
and then to a four-volume in spacetime. The figure connects these geometric objects to
the features used in QGrAT, where $P$ is the summed constituent four-momentum,
$T_i$ is the weighted trivector sum defined in Section~\ref{sec:ising_invariance},
$h_i$ is the weighted four-volume magnitude, and $p_{\text{lead}}$ denotes the
four-momentum of the leading-$p_T$ constituent.

\begin{figure}[H]
\centering
\begin{tikzpicture}[
  x=1cm, y=1cm,
  line join=round, line cap=round,
  grade edge/.style={line width=0.7pt},
  grade arrow/.style={grade edge,-{Stealth[length=1.7mm,width=1.2mm]}},
  grade label/.style={font=\footnotesize,inner sep=1pt,align=center}
]
  \foreach \x/\g/\name/\colour/\ncomp in {
    0/0/scalar/gray/1,
    3.2/1/vector/blue/4,
    6.4/2/bivector/purple/6,
    9.6/3/trivector/orange/4,
    12.8/4/pseudoscalar/red/1}{
    \begin{scope}[shift={(\x,0)}]
      \draw[draw=\colour!45,fill=\colour!3,line width=0.4pt,rounded corners=2pt]
        (-1.51,0) rectangle (1.51,3.35);
      \node[grade label,text=\colour!65!black] at (0,3.09) {grade $\g$};
      \node[grade label,font=\small\bfseries] at (0,2.76) {\name};
      \draw[\colour!30,line width=0.4pt] (-1.28,2.5) -- (1.28,2.5);
      \node[grade label,font=\scriptsize] at (0,0.23)
        {\ncomp\ basis component\ifnum\ncomp=1\else s\fi};
    \end{scope}
  }

  \fill[gray!75] (-0.3,1.73) circle (3pt);
  \node[grade label,text=gray!65!black] at (0.25,1.73) {$s\in\RR$};
  \node[grade label] at (0,0.67) {$p_i\cdot P$};

  \begin{scope}[shift={(3.2,0)}]
    \draw[grade arrow,blue!75!black] (-0.85,1.3) -- (0.85,2.13);
    \node[grade label] at (0,0.67) {$p_i$};
  \end{scope}

  \begin{scope}[shift={(6.4,0)}]
    \path[draw=purple!80!black,fill=purple!16,grade edge]
      (-1.02,1.22) -- (0.48,1.22) -- (1.02,2.12) -- (-0.48,2.12) -- cycle;
    \draw[grade arrow,purple!80!black]
      (0.28,1.57) arc[start angle=-35,end angle=255,x radius=0.36,y radius=0.24];
    \node[grade label] at (0,0.67) {$p_i\wedge P$};
  \end{scope}

  \begin{scope}[shift={(9.6,0)},draw=orange!75!black]
    \draw[fill=orange!22,grade edge]
      (-0.82,1.1) -- (0.37,1.1) -- (0.37,2.0) -- (-0.82,2.0) -- cycle;
    \draw[fill=orange!10,grade edge]
      (-0.82,2.0) -- (-0.37,2.36) -- (0.82,2.36) -- (0.37,2.0) -- cycle;
    \draw[fill=orange!32,grade edge]
      (0.37,1.1) -- (0.82,1.46) -- (0.82,2.36) -- (0.37,2.0) -- cycle;
    \node[grade label,text=orange!75!black] at (-0.2,1.58) {$+$};
    \node[grade label] at (0,0.67) {$T_i$};
  \end{scope}

  \begin{scope}[shift={(12.8,1.73)},draw=red!70!black]
    \coordinate (A) at (-0.9,-0.63);
    \coordinate (B) at (0.45,-0.63);
    \coordinate (C) at (0.45,0.34);
    \coordinate (D) at (-0.9,0.34);
    \coordinate (E) at (-0.45,-0.3);
    \coordinate (F) at (0.9,-0.3);
    \coordinate (G) at (0.9,0.67);
    \coordinate (H) at (-0.45,0.67);
    \path[fill=red!5] (A) -- (B) -- (F) -- (G) -- (H) -- (D) -- cycle;
    \draw[grade edge,red!45] (A) -- (B) -- (C) -- (D) -- cycle
      (E) -- (F) -- (G) -- (H) -- cycle
      (A) -- (E) (B) -- (F) (C) -- (G) (D) -- (H);
    \foreach \v in {A,B,C,D,E,F,G,H}{
      \path (\v) coordinate (outer\v);
      \path (0,0) -- (\v) coordinate[pos=0.47] (inner\v);
      \draw[red!35,line width=0.4pt] (outer\v) -- (inner\v);
    }
    \draw[grade edge] (innerA) -- (innerB) -- (innerC) -- (innerD) -- cycle
      (innerE) -- (innerF) -- (innerG) -- (innerH) -- cycle
      (innerA) -- (innerE) (innerB) -- (innerF)
      (innerC) -- (innerG) (innerD) -- (innerH);
    \node[grade label,text=red!70!black,font=\scriptsize] at (0,0) {$+$};
    \node[grade label] at (0,-1.06) {$h_i$};
  \end{scope}
\end{tikzpicture}
\caption{Geometric interpretation of the five grades in $\Cl(1,3)$ and the corresponding
QGrAT features. The drawings schematically represent simple spacetime blades, with grade~4
shown as a projection of an oriented four-volume. The $+$ signs select an orientation
for the volumes. The component counts are the grade-subspace dimensions $\binom{4}{g}$,
giving $1+4+6+4+1=16$ basis blades. The implemented $T_i$ sums weighted
trivectors, while $h_i$ sums four-volume magnitudes and retains no orientation sign.
All contractions use the Minkowski metric
$\mathrm{diag}(+1,-1,-1,-1)$.}
\label{fig:clifford_grades}
\end{figure}
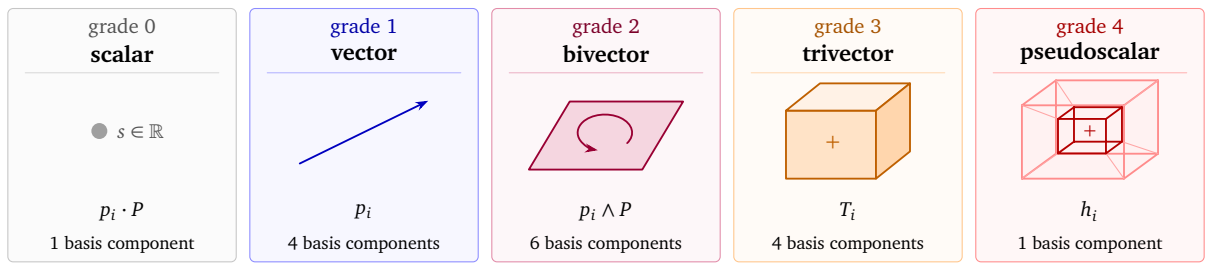

\subsection{Lorentz Equivariance via the Spin Group}
\label{sec:spin_group}

The proper orthochronous Lorentz group $\mathrm{SO}^+(1,3)$ has as its double cover the
identity component $\mathrm{Spin}^+(1,3)\cong\mathrm{SL}(2,\CC)$, every element of which
satisfies $s\,s^\dagger=1$. It acts on multivectors through the twisted
adjoint representation~\cite{doran_lasenby,gatr},
\begin{equation}
  \rho(s)\colon \mv{x} \;\mapsto\; s\,\mv{x}\,s^\dagger,
  \qquad s\in\mathrm{Spin}^+(1,3)\subset\Cl(1,3).
  \label{eq:spin_action}
\end{equation}
This action preserves each grade subspace, so grade projection commutes with it,
$\grade{s\mv{x}s^\dagger}{k} = s\,\grade{\mv{x}}{k}\,s^\dagger$ for all $k$~\cite[Ch.~4]{doran_lasenby}.
L-GATr~\cite{lgatr,gatr} constructs all of its layers as $\mathrm{Spin}^+(1,3)$-equivariant
maps under this action, guaranteeing that the entire forward pass commutes with Lorentz
transformations.

We do not implement this group action on the qubits. The Lorentz group is non-compact and
has no nontrivial finite-dimensional unitary representation, so a unitary circuit cannot
carry it. Instead we use the quantities that are invariant under this action, the Clifford
inner products of Section~\ref{sec:equivariance}, as the circuit's couplings.

\subsection{Pauli Representation of Clifford Algebra}
\label{sec:correspondence}

The Euclidean Clifford algebra $\Cl(3,0)$, generated by $\{e_1,e_2,e_3\}$ with
$e_ie_j+e_je_i=2\delta_{ij}\mathbb{1}$, is isomorphic to $\mathrm{End}(\CC^2)$ via the
identification $e_i\leftrightarrow\sigma_i$ with the Pauli matrices~\cite{doran_lasenby,trindade_clifford_qnn}.
For the Minkowski algebra $\Cl(1,3)$, the Dirac representation gives a $4\times4$ matrix
representation as tensor products of Pauli operators. This correspondence provides the
algebraic motivation for expressing our circuit in terms of Pauli rotations. In our
construction, multivector components determine single-qubit rotation angles, while Clifford
inner products determine the angles of two-qubit gates in the grade-aware attention.
A grade subspace can contain several blades (six for the bivector grade and four for the
trivector grade), so the encoder of Section~\ref{sec:encoding_layer} assigns a distinct rotation
to each blade component.

To understand how a pairwise relationship enters the attention, consider two particles and
one multivector grade. Their Clifford inner product gives a single real number $\eta_g$,
which we use directly as the angle of an IsingZZ gate,
$\mathrm{IsingZZ}(\eta_g)=\exp(-\tfrac{i}{2}\,\eta_g\,Z_iZ_j)$.
This gate changes the relative phase between the components of the quantum state with
matching qubit values, $\ket{00}$ and $\ket{11}$, and those with different values,
$\ket{01}$ and $\ket{10}$. It therefore records the particle relationship without immediately
changing the probabilities of these four outcomes. The phase determines how the components
interfere when subsequently mixed, and rotations after the attention gates allow this information
to affect the measured classification output. Encoding data through such ZZ phases is an
established quantum feature-map construction~\cite{havlicek_feature_map}. It gives us a simple
way to introduce each pairwise scalar into the circuit, with no independent trainable
parameter in the attention coupling.

The Lorentz invariance of this coupling comes from the scalar $\eta_g$, rather than the
choice of the ZZ operator. An IsingXX or IsingYY gate driven by the same invariant angle
would retain that property. Our choice of ZZ is thus motivated by its simple phase-encoding
interpretation. It is neither required by Clifford algebra nor established here to be superior
to these alternatives. This choice applies specifically to the grade-aware attention. The
quantum readout uses both IsingZZ and IsingYY gates with trainable angles to collect information
from the data and ancilla qubits (Section~\ref{sec:qgatqr_readout}). The Clifford--Pauli
correspondence motivates this construction but does not by itself make the component encoding
Lorentz covariant.

\section{Lorentz Symmetry of the Attention Couplings}
\label{sec:equivariance}

Here, we briefly formalise a key structural property of grade-aware attention. Clifford inner
products are invariant when their multivector arguments transform together under the Lorentz
group. Figure~\ref{fig:clifford_grades} illustrates the grade-wise features whose pairings
provide the attention couplings. We first establish this invariance and then specify its scope for the
implemented features, which are constructed from constituent momenta at every grade.
Constituent selection and the leading-particle reference restrict which transformations
of a jet leave the resulting couplings unchanged
(Section~\ref{sec:ising_invariance}).

\subsection{Clifford Inner Products Are Lorentz Scalars}
\label{sec:ip_scalar}

The Lorentz invariance of the scalar (grade-0) part of a geometric product is a foundational
property of spacetime geometric algebra, established by~\cite{hestenes_sta} and developed
in the modern treatments of~\cite{doran_lasenby} and~\cite{lounesto_clifford}.
The same invariance underlies the equivariant attention scores used in
GATr~\cite{gatr} and L-GATr~\cite{lgatr}.
We recall the explicit statement below for the grade-$g$ inner products that will serve as
$\mathrm{IsingZZ}$ entanglement angles in our circuit, both for self-containedness and to
fix the form used in the subsequent QGrAT-specific result.

\begin{lemma}[Lorentz invariance of Clifford inner products~\cite{hestenes_sta,doran_lasenby,lounesto_clifford,gatr,lgatr}]
\label{lem:lorentz_invariance}
Let $\mv{x}_i,\mv{x}_j\in\Cl(1,3)$ be two multivectors and $s\in\mathrm{Spin}^+(1,3)$.
Define the grade-$g$ Clifford inner product
\begin{equation}
  \eta_g^{(ij)}
  = \bigl\langle \grade{\mv{x}_i}{g}^{\dagger}\;\grade{\mv{x}_j}{g} \bigr\rangle_0
  \;\in\;\RR,
  \label{eq:clifford_ip}
\end{equation}
where $(\cdot)^\dagger$ is the Clifford reverse and $\langle\cdot\rangle_0$ keeps only the
scalar (grade-0) part of a multivector.
Then $\eta_g^{(ij)}$ is invariant under a simultaneous Lorentz transformation of both
multivectors,
\begin{equation}
  \eta_g\bigl(\rho(s)\mv{x}_i,\;\rho(s)\mv{x}_j\bigr)
  = \eta_g^{(ij)}.
  \label{eq:invariance}
\end{equation}
Here $\rho(s)\mv{x}_i$ denotes the multivector $\mv{x}_i$ after the Lorentz transformation
represented by $s$ acts on it through the sandwich map~\eqref{eq:spin_action}. Equation~\eqref{eq:invariance}
therefore states that the inner product computed in the transformed frame equals the one
computed in the original frame.
\end{lemma}

\begin{proof}
The derivation is standard~\cite[Ch.~5]{doran_lasenby} and reproduced here for
self-containedness.
Under the twisted adjoint action~\eqref{eq:spin_action}, each grade subspace is
preserved~\cite{doran_lasenby,gatr}, so
$\grade{\rho(s)\mv{x}}{g} = s\,\grade{\mv{x}}{g}\,s^\dagger$.
Because the reverse flips the order of a product, $(XY)^\dagger = Y^\dagger X^\dagger$, and
applying it twice returns the original, $(A^\dagger)^\dagger = A$, we have
$(s\,A\,s^\dagger)^\dagger = s\,A^\dagger\,s^\dagger$ for any multivector $A$~\cite{lounesto_clifford}.
We obtain
\begin{align}
  \eta_g\bigl(\rho(s)\mv{x}_i,\,\rho(s)\mv{x}_j\bigr)
  &= \bigl\langle
       (s\,\grade{\mv{x}_i}{g}\,s^\dagger)^\dagger
       \;(s\,\grade{\mv{x}_j}{g}\,s^\dagger)
     \bigr\rangle_0 \notag\\
  &= \bigl\langle
       s\,\grade{\mv{x}_i}{g}^\dagger\,s^\dagger\,
       s\,\grade{\mv{x}_j}{g}\,s^\dagger
     \bigr\rangle_0
   = \bigl\langle
       s\,\grade{\mv{x}_i}{g}^\dagger\,\grade{\mv{x}_j}{g}\,s^\dagger
     \bigr\rangle_0.
\end{align}
The grade-0 projection of $s\,M\,s^\dagger$ equals the grade-0 projection of $M$ for any
$s\in\mathrm{Spin}^+(1,3)$ and any multivector $M$, since the scalar part is invariant under
the spin-group action~\cite{doran_lasenby,gatr}.\\
Therefore
$\eta_g(\rho(s)\mv{x}_i,\rho(s)\mv{x}_j) = \langle\grade{\mv{x}_i}{g}^\dagger\,\grade{\mv{x}_j}{g}\rangle_0 = \eta_g^{(ij)}$.
\end{proof}

\subsection{Grade-1 Specialisation to the Minkowski Metric}
\label{sec:minkowski}

The identification of the scalar part of the geometric product of two four-vectors with
the Minkowski inner product is a standard result in geometric algebra, established in the
foundational treatment of spacetime algebra by~\cite{hestenes_sta} and presented in
modern form by~\cite{doran_lasenby} and~\cite{lounesto_clifford}.
It also underpins the equivariant attention scores used in the classical L-GATr
architecture~\cite{lgatr,gatr}, where pairwise Minkowski products of grade-1 components
serve as Lorentz-invariant scores.
We recall the identity below in the form used by QGrAT, both for self-containedness and to
make explicit how the grade-1 IsingZZ angle in our circuit reduces to this canonical
Lorentz invariant.

\begin{lemma}[Grade-1 inner product as Minkowski metric~\cite{hestenes_sta,doran_lasenby,lounesto_clifford}]
\label{lem:minkowski}
Let $\mv{p}_i=p_i^\mu e_\mu$ and $\mv{p}_j=p_j^\mu e_\mu$ be grade-1 multivectors encoding
four-momenta.
Then
\begin{equation}
  \eta_1^{(ij)}
  = \bigl\langle\mv{p}_i^\dagger\,\mv{p}_j\bigr\rangle_0
  = E_iE_j - p_x^{(i)}p_x^{(j)} - p_y^{(i)}p_y^{(j)} - p_z^{(i)}p_z^{(j)}
  = p_i^\mu\eta_{\mu\nu}p_j^\nu.
  \label{eq:minkowski_ip}
\end{equation}
\end{lemma}

\begin{proof}
The derivation follows the standard one in~\cite[Ch.~5]{doran_lasenby} and is reproduced
here for the explicit form of the grade-1 case used in this paper.
Grade-1 multivectors are reverse-invariant ($\mv{p}^\dagger=\mv{p}$).
Expanding the geometric product symmetrically and antisymmetrically using the defining
relation~\eqref{eq:clifford_relation} gives
\begin{equation}
  \mv{p}_i\mv{p}_j
  = p_i^\mu p_j^\nu e_\mu e_\nu
  = \underbrace{p_i^\mu p_j^\nu \eta_{\mu\nu}\,\mathbb{1}}_{\text{grade 0}}
   + \underbrace{p_i^\mu p_j^\nu\,\tfrac{1}{2}[e_\mu,e_\nu]}_{\text{grade 2}}.
\end{equation}
Projecting onto grade 0 retains only the first term, giving
$\langle\mv{p}_i\mv{p}_j\rangle_0 = p_i^\mu\eta_{\mu\nu}p_j^\nu$.
\end{proof}

Lemma~\ref{lem:minkowski} establishes that the grade-1 Clifford inner product is exactly
the Minkowski dot product, which is the fundamental Lorentz invariant of special
relativity~\cite{doran_lasenby,lorentznet}.
The QGrAT circuit therefore uses the Minkowski inner product between particle pairs as its
grade-1 IsingZZ entanglement angle, and analogous Clifford pairings for the higher grades.
The invariance of the latter under transformations of the jet depends on the reference-vector
prescription specified below.
This generalises the explicit-invariant approach of LorentzNet~\cite{lorentznet}
to the full Clifford grade decomposition.

\subsection{Invariance of Entangling-Gate Coefficients and Its Scope}
\label{sec:ising_invariance}

\begin{lemma}[Invariance of the entangling-gate coefficients under simultaneous Lorentz transformations]
\label{prop:circuit_equivariance}
Let $U_{\text{ent}}(\{\eta_g^{(ij)}\},\theta)$ denote the entangling block of the grade-aware
attention, built from the pairwise $\mathrm{IsingZZ}$ gates whose angles are the Clifford inner
products $\eta_g^{(ij)}$, together with the data-independent variational parameters $\theta$.
This block excludes the per-token encoding rotations, which are treated separately.
Let $\Lambda\in\mathrm{SO}^+(1,3)$ be a Lorentz transformation with spin-group representative
$s\in\mathrm{Spin}^+(1,3)$.
Fix the constituent labels and gate order, and suppose every multivector argument transforms
as $\mv{x}_i\mapsto\rho(s)\mv{x}_i$. For features constructed from several four-vectors,
this requires all constituent and reference vectors entering the feature to co-transform.
Then the entangling-gate angles are unchanged,
\begin{equation}
  U_{\text{ent}}\bigl(\{\eta_g(\rho(s)\mv{x}_i,\rho(s)\mv{x}_j)\},\,\theta\bigr)
  \;=\;
  U_{\text{ent}}\bigl(\{\eta_g^{(ij)}\},\,\theta\bigr),
  \label{eq:circuit_equivariance}
\end{equation}
so the operator built from these coefficients is the same before and after the simultaneous
transformation of its multivector arguments. This statement does not assert invariance under
transforming the constituents while holding arbitrary reference vectors fixed.
\end{lemma}

\begin{proof}
Each $\mathrm{IsingZZ}$ gate in $U_{\text{ent}}$ has angle $\eta_g^{(ij)}$ for some grade $g$ and
pair $(i,j)$. By Lemma~\ref{lem:lorentz_invariance}, every such angle is unchanged under
$\mv{x}\mapsto\rho(s)\mv{x}$. The variational angles $\theta$ are independent of the input data
and so are unchanged as well. The operator built from these angles is therefore identical before
and after the transformation.
\end{proof}

The scope of this invariance depends on the constituent selection and the reference vectors
entering each grade. First fix the identities of the constituents represented in the circuit
and of those included in the sum $P=\sum_k p_k$. Transforming each four-momentum then gives
$P\mapsto\Lambda P$. The scalar feature $p_i\cdot P$, the grade-1 momentum $p_i$, and the
grade-2 bivector $p_i\wedge P$ consequently satisfy the transformation assumptions of
Lemma~\ref{lem:lorentz_invariance}. Their grade-wise pairings are invariant under the full
proper orthochronous Lorentz group for this fixed constituent set.

The higher-grade features used in this study are constructed from the retained $N$
constituent four-momenta, with $P=\sum_{k=1}^{N}p_k$. Defining the pairwise Minkowski
inner products $d_{ij}=p_i\cdot p_j$, we form
\begin{equation}
 T_i=\sum_j d_{ij}\,(p_i\wedge P\wedge p_j),\qquad
 h_i=\sum_j d_{ij}\,\left|\det[p_i,P,p_j,p_{\text{lead}}]\right|,
 \label{eq:reference_blades}
\end{equation}
where $p_{\text{lead}}$ is the four-momentum of the retained constituent with the highest
transverse momentum. The wedge products in $T_i$ describe oriented three-dimensional
volume elements in spacetime, weighted by the invariant relationships between particle
pairs. For a fixed constituent set, their sum transforms as a trivector under the full
proper orthochronous Lorentz group. The determinant in $h_i$ similarly describes an
oriented four-dimensional volume, but taking its absolute value removes the orientation
information. The resulting magnitude is therefore parity even, whereas an oriented
pseudoscalar changes sign under parity. We use this magnitude as the feature associated
with grade~4 and define its attention coupling as $\eta_4^{(ij)}=-h_i h_j$.
This coupling is invariant when the selected leading constituent transforms together with
the other jet constituents.

We adopt the component signatures $(+)$, $(+,-,-,-)$,
$(+,+,+,-,-,-)$, $(-,-,-,+)$ and $(-)$ for grades~0 through~4.
For grades~2 and~3 these are the negatives of the reverse-based pairing in
Eq.~\eqref{eq:clifford_ip}. Accordingly, the attention coefficients are
$\eta_g^{(ij)}=\epsilon_g\langle\grade{\mv{x}_i}{g}^{\dagger}
\grade{\mv{x}_j}{g}\rangle_0$, with $\epsilon_2=\epsilon_3=-1$ and
$\epsilon_0=\epsilon_1=\epsilon_4=1$, interpreting the grade-4 feature as just described.
These fixed sign choices preserve the invariance argument and specify the gate angles
used throughout this study.

Constituent selection is a separate limitation. The top-$N$ and leading-particle selections
by transverse momentum are preserved by azimuthal rotations and longitudinal boosts, but
need not be preserved by a general Lorentz transformation. If the constituents are reselected
after such a transformation, the gate inputs can change even at grades~0--2. If the set
entering $P$ changes, it need not equal the Lorentz transform of the original sum either.
Thus the full-group invariance of the lower-grade pairings is a statement about a fixed
constituent set, not about the complete selection and encoding pipeline.

Lemma~\ref{prop:circuit_equivariance} is a statement about the gate coefficients, not about the
circuit state. It establishes that the pattern of pairwise couplings is set by the Clifford
pairings, with the transformation properties just described, but it does not
say that the encoded state or its entanglement is frame independent. The per-token encoding rotations $R_X,R_Y,R_Z$ load the
individual multivector components of each constituent, and those components change from frame to
frame, so the input state that the entangling block acts on is itself frame dependent. An unchanged
entangling gate acting on a changed input state produces a different entangled state, and this
frame dependence applies to the component encoding of every grade, not only the grade-1
four-momentum. The complete classifier is therefore frame dependent and is not an equivariant map.
A fully equivariant quantum circuit, in which both the encoding and the variational stages transform
covariantly under $\mathrm{Spin}^+(1,3)$, would need a different encoding and is left to future work.

\section{Methodology}
\label{sec:method}

\subsection{Overview}
\label{sec:arch_overview}

Building on the geometric features and symmetry properties established above, we now
explain how QGrAT turns the constituent momenta of a jet into a classification decision.
We want the circuit to use both the features of individual constituents and the relationships
between them. Figure~\ref{fig:arch_overview} follows this flow from the classical preparation
of the inputs to the final readout. Each retained constituent is assigned a data qubit
$q_i$, giving a register $q_0\dots q_{N-1}$. Auxiliary qubits then provide a way to introduce
an explicit three-particle dependence and, in the main model, collect the classification signal.

To bring the jet information into this register, we first prepare two sets of quantities
from the constituent multivectors: rotation angles $\theta_i$ that load the individual data
qubits, and Clifford inner products $\eta_g^{(ij)}$ that determine how pairs of qubits
interact. This separates the encoding of constituent features from the geometric relationships
used to set the attention couplings. We process each active grade in turn, encoding its blade
components through single-qubit rotations (Section~\ref{sec:encoding_layer}) and applying
$\mathrm{IsingZZ}(\eta_g^{(ij)})$ gates between particle pairs. Four-momentum information is
always present through grade~1, whose pairing is the Minkowski dot product
(Lemma~\ref{lem:minkowski}). Alongside it, scalar, bivector and trivector features introduce
other geometric information, while the grade-4 four-volume magnitude provides an additional
feature whose effect we assess separately.

Beyond these pairwise couplings, we ask whether explicitly introducing a product involving
three constituents helps the circuit recognise the structure of a jet. For this purpose,
we add an ancilla $q_T$ that controls phase rotations over triples of data qubits, with angles
set by the product of their three pairwise Minkowski inner products. This construction uses
the principle of the Hadamard test~\cite{cleve1998quantum,nielsen2010quantum,ekert2002direct}
and carries the resulting correlations forward without an intermediate measurement
(Section~\ref{sec:injected_scalars}). A shallow variational (VQC) block then mixes the data
and ancilla qubits, allowing the model to learn how to combine the information supplied by
the attention and the three-particle product (Section~\ref{sec:vqc_layer}).

Once this state has been prepared, we need to extract a classification score. We compare two
ways of doing so while keeping the architecture of the preceding circuit fixed. In the main model, a readout
qubit $q_r$ interacts with the data and ancilla qubits, and its measured expectation value
determines the classification logit (Section~\ref{sec:qgatqr_readout}). Alternatively, we
measure pairwise Pauli correlations across the data and ancilla register and combine them
in a learned classical Hamiltonian-observable head (Section~\ref{sec:qgat_readout}). This
comparison lets us assess how the readout choice affects classification within a common
quantum architecture.

Alongside this readout comparison, we vary the active grades, the number of constituents and
the presence of the three-particle-product ancilla to establish what each choice contributes
to the performance reported in Section~\ref{sec:results}. The following subsections develop
the individual stages in detail. To make their gate structure easy to follow,
Figs.~\ref{fig:block_attention}--\ref{fig:block_tap} use an illustrative register of four data
qubits and a single attention layer ($L=1$). The same construction extends to the
$N\in\{8,10\}$ constituents used in our experiments (Section~\ref{sec:experimental_setup}).

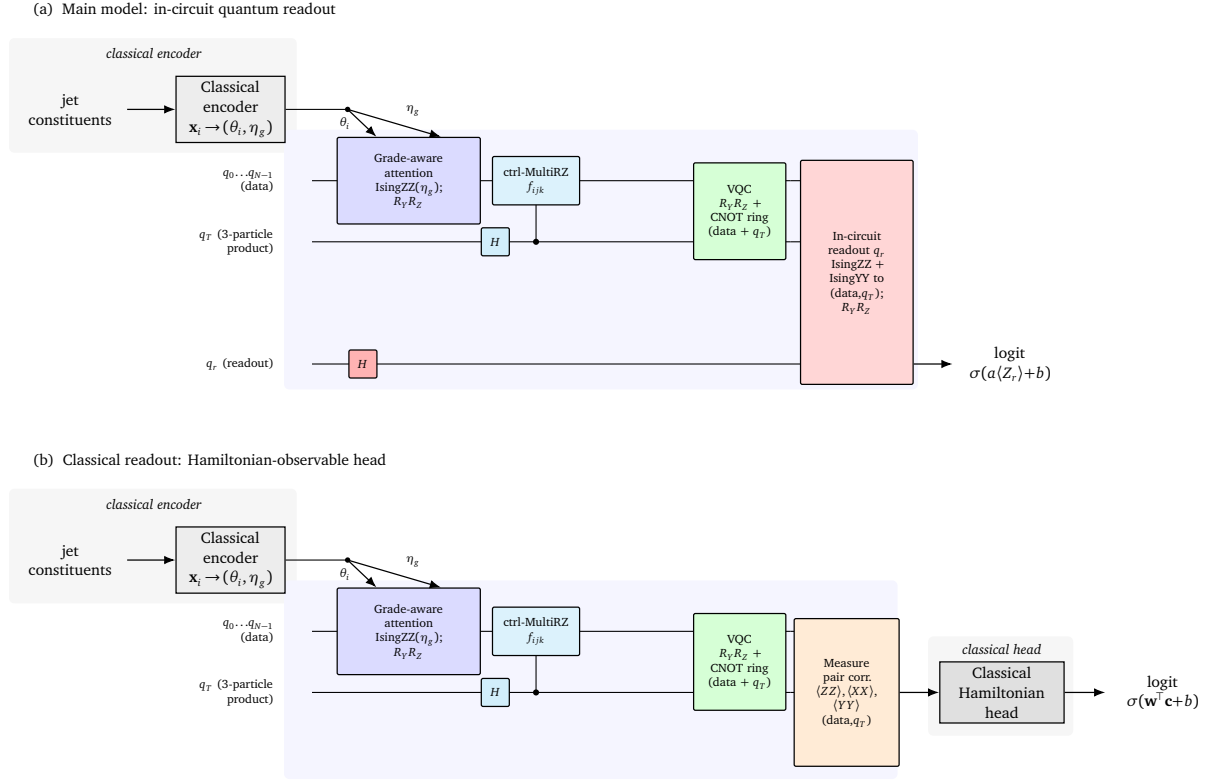
\begin{figure}[t]
\centering
\resizebox{\linewidth}{!}{%
\begin{tikzpicture}
  \begin{scope}[yshift=0cm]
    \fill[gray!7, rounded corners=3pt]  (-1.45,0.55) rectangle (4.2,2.8);
    \fill[blue!4, rounded corners=3pt]  (3.95,-4.1) rectangle (16.4,1.0);
    \node[hdreg] at (1.4,2.5) {classical encoder};
    \node[font=\hdfont\small, anchor=west] at (-1.1,3.35) {(a)\ \ Main model: in-circuit quantum readout};
    \qgatfrontend{a}
    \node[hdlane] at (3.85,-3.6) {$q_r$ (readout)};
    \draw[hdwire] (4.5,-3.6) -- (14.16,-3.6);
    \node[hdgate={red!30}, minimum width=0.55cm, minimum height=0.55cm] (aHr) at (5.5,-3.6) {$H$};
    \draw[hdwire] (13.9,0) -- (14.16,0);
    \draw[hdwire] (13.9,-1.2) -- (14.16,-1.2);
    \node[hdgate={red!16}, text width=2.0cm, minimum height=4.4cm] (aro) at (15.2,-1.8)
         {In-circuit\\ readout $q_r$\\ IsingZZ ${+}$\\ IsingYY to\\ (data,$q_T$);\\ $R_YR_Z$};
    \node[hdio] (aout) at (18.2,-3.6) {logit\\ $\sigma(a\langle Z_r\rangle{+}b)$};
    \draw[hdarr] (aro.east|-aout) -- (aout.west);
  \end{scope}
  \begin{scope}[yshift=-8.85cm]
    \fill[gray!7, rounded corners=3pt]  (-1.45,0.55) rectangle (4.2,2.8);
    \fill[blue!4, rounded corners=3pt]  (3.95,-2.9) rectangle (16.0,1.0);
    \fill[gray!7, rounded corners=3pt]  (16.6,-2.05) rectangle (19.45,-0.1);
    \node[hdreg] at (1.4,2.5) {classical encoder};
    \node[hdreg] at (18.05,-0.35) {classical head};
    \node[font=\hdfont\small, anchor=west] at (-1.1,3.35) {(b)\ \ Classical readout: Hamiltonian-observable head};
    \qgatfrontend{c}
    \draw[hdwire] (13.9,0) -- (14.06,0);
    \draw[hdwire] (13.9,-1.2) -- (14.06,-1.2);
    \node[hdgate={orange!16}, text width=1.85cm, minimum height=2.9cm] (ccorr) at (15.0,-1.2)
         {Measure\\ pair corr.\\ $\langle ZZ\rangle,\langle XX\rangle,$\\ $\langle YY\rangle$\\ (data,$q_T$)};
    \node[hdblock={gray!24}, text width=2.3cm, minimum height=1.2cm] (cham) at (18.05,-1.2)
         {Classical\\ Hamiltonian\\ head};
    \node[hdio] (cout) at (21.2,-1.2) {logit\\ $\sigma(\mv{w}^{\!\top}\mv{c}{+}b)$};
    \draw[hdarr] (ccorr.east) -- (cham.west);
    \draw[hdarr] (cham.east) -- (cout.west);
  \end{scope}
\end{tikzpicture}%
}
\caption{The two QGrAT quantum circuits studied in this work. Both share the classical
encoder, grade-aware $\mathrm{IsingZZ}(\eta_g)$ attention, the three-particle-product ancilla
tap, and the variational (VQC) mixer, and differ only in the readout stage. \textbf{(a)} Main
model, an in-circuit readout qubit $q_r$ coupled to the data register and the
three-particle-product ancilla $q_T$. \textbf{(b)} Classical-readout variant, where $q_r$ is
removed and Pauli pair correlators over the data and ancilla qubits feed a learned
Hamiltonian-observable head. Shaded boxes are classical stages drawn only as the source of the
circuit inputs. Gate-level building blocks appear in
Figs.~\ref{fig:block_attention}--\ref{fig:block_tap}.}
\label{fig:arch_overview}
\end{figure}

\subsection{Encoding Multivectors into Rotation Angles}
\label{sec:encoding_layer}

The encoding layer maps each token's 16-component Clifford multivector $\mv{x}_i\in\Cl(1,3)$
to a vector of rotation angles via a learned encoder
$\mathcal{E}:\Cl(1,3)\to\RR^{n_\alpha}$.
Following the grade decomposition of Section~\ref{sec:clifford}, the encoder treats each
grade subspace separately to preserve algebraic structure.

For grades 0 and 4 (scalar and four-volume magnitude, one component each), the corresponding
component is multiplied by a learned scalar and passed through $\arctan$, producing a single
rotation angle.
For grade 1 (four-momentum, four dimensions $E,p_x,p_y,p_z$), each of the four components is
first multiplied by its own learned scale and then passed through $\arctan$, giving four
angles that are realised as the gate sequence
$R_X(\alpha^E)\,R_Y(\alpha^{p_x})\,R_Z(\alpha^{p_y})\,R_X(\alpha^{p_z})$ on the qubit.
The four grade-1 scales form a trainable vector $\mathbf{s}^{(1)}\in\RR^4$ (one gain per
momentum component). The encoded angle for blade $b$ is
$\alpha^{(1)}_b=\arctan\!\big(s^{(1)}_b\,x^{(1)}_b\big)$, where $x^{(1)}_b$ is the raw
four-momentum component.

Grades 2 and 3 are encoded in exactly the same way, and this is the direct higher-grade
analogue of the grade-1 scales. The grade-2 (bivector) blade is the six-component wedge
product $p_i\wedge P$ of the token against the total jet four-momentum $P=\sum_k p_k$, and
the grade-3 feature is the four-component weighted trivector $T_i$ of
Eq.~\eqref{eq:reference_blades}, with the transformation behaviour of $P$ and
$p_{\text{lead}}$ set out in Section~\ref{sec:ising_invariance}. Each grade carries its own trainable
scale vector, $\mathbf{s}^{(2)}\in\RR^6$ for the six bivector blades and
$\mathbf{s}^{(3)}\in\RR^4$ for the four trivector blades,
playing the same per-blade role that $\mathbf{s}^{(1)}$ plays for the four-momentum.
Every blade is turned into an angle by the same rule
$\alpha^{(g)}_b=\arctan\!\big(s^{(g)}_b\,x^{(g)}_b\big)$, and the angles are applied as
single-qubit rotations on the token's qubit, the six bivector angles as
$R_X\,R_Y\,R_Z\,R_X\,R_Y\,R_Z$ and the four trivector angles as $R_X\,R_Y\,R_Z\,R_X$.
Every supplied component affects the feature map, but successive rotations on one qubit
do not imply an injective encoding of all six or four coordinates.
Across the active grades $0,1,2,3$ the encoder therefore contributes $1+4+6+4=15$ learned
scales, producing $15$ rotation angles per token, all trained via backpropagation.

The $\arctan$ activation keeps each encoded rotation angle in the bounded interval
$(-\pi/2,\pi/2)$.
The component-wise encoding is frame dependent at every grade, because each blade component is
loaded as its own rotation angle and these components change from frame to frame, as discussed in
Section~\ref{sec:ising_invariance}.

\subsection{Grade-Aware Attention Layer}
\label{sec:attn_layer}

The attention layer operates on $N$ data qubits, one per jet constituent.
The full multi-grade block is drawn in Fig.~\ref{fig:block_attention}.
For each attention layer $\ell\in\{1,\ldots,L\}$, the block is applied
in ascending order of active grade to all data qubits. The grade-1 example below
illustrates the three operations repeated for every active grade.
The attention block consists of three sub-operations.

\begin{figure}[H]
\centering
\resizebox{\textwidth}{1.1cm}{%
\begin{quantikz}[column sep=0.20cm, row sep=0.18cm]
  \lstick{$q_0$}
    & \gate[style={fill=gray!18}]{R_Y(\alpha_0)}
    & \gate[wires=4,style={fill=gray!30}]{\mathrm{ZZ}_0}
    & \gate[style={fill=green!18}]{R_Y R_Z}
    & \gate[style={fill=blue!12}]{R_X}
    & \gate[style={fill=blue!22}]{R_Y}
    & \gate[style={fill=blue!32}]{R_Z}
    & \gate[style={fill=blue!42}]{R_X}
    & \gate[wires=4,style={fill=orange!25}]{\mathrm{ZZ}_1}
    & \gate[style={fill=green!25}]{R_Y R_Z}
    & \gate[style={fill=purple!12}]{R_X}
    & \gate[style={fill=purple!20}]{R_Y}
    & \gate[style={fill=purple!28}]{R_Z}
    & \gate[style={fill=purple!36}]{R_X}
    & \gate[style={fill=purple!44}]{R_Y}
    & \gate[style={fill=purple!52}]{R_Z}
    & \gate[wires=4,style={fill=purple!30}]{\mathrm{ZZ}_2}
    & \gate[style={fill=green!32}]{R_Y R_Z}
    & \gate[style={fill=orange!18}]{R_X}
    & \gate[style={fill=orange!26}]{R_Y}
    & \gate[style={fill=orange!34}]{R_Z}
    & \gate[style={fill=orange!42}]{R_X}
    & \gate[wires=4,style={fill=orange!40}]{\mathrm{ZZ}_3}
    & \gate[style={fill=green!38}]{R_Y R_Z}
    & \gate[style={fill=red!20}]{R_Y}
    & \gate[wires=4,style={fill=red!38}]{\mathrm{ZZ}_4}
    & \gate[style={fill=green!45}]{R_Y R_Z}
    & \qw \\
  \lstick{$q_1$}
    & \gate[style={fill=gray!18}]{R_Y}
    &
    & \gate[style={fill=green!18}]{R_Y R_Z}
    & \gate[style={fill=blue!12}]{R_X}
    & \gate[style={fill=blue!22}]{R_Y}
    & \gate[style={fill=blue!32}]{R_Z}
    & \gate[style={fill=blue!42}]{R_X}
    &
    & \gate[style={fill=green!25}]{R_Y R_Z}
    & \gate[style={fill=purple!12}]{R_X}
    & \gate[style={fill=purple!20}]{R_Y}
    & \gate[style={fill=purple!28}]{R_Z}
    & \gate[style={fill=purple!36}]{R_X}
    & \gate[style={fill=purple!44}]{R_Y}
    & \gate[style={fill=purple!52}]{R_Z}
    &
    & \gate[style={fill=green!32}]{R_Y R_Z}
    & \gate[style={fill=orange!18}]{R_X}
    & \gate[style={fill=orange!26}]{R_Y}
    & \gate[style={fill=orange!34}]{R_Z}
    & \gate[style={fill=orange!42}]{R_X}
    &
    & \gate[style={fill=green!38}]{R_Y R_Z}
    & \gate[style={fill=red!20}]{R_Y}
    &
    & \gate[style={fill=green!45}]{R_Y R_Z}
    & \qw \\
  \lstick{$q_2$}
    & \gate[style={fill=gray!18}]{R_Y}
    &
    & \gate[style={fill=green!18}]{R_Y R_Z}
    & \gate[style={fill=blue!12}]{R_X}
    & \gate[style={fill=blue!22}]{R_Y}
    & \gate[style={fill=blue!32}]{R_Z}
    & \gate[style={fill=blue!42}]{R_X}
    &
    & \gate[style={fill=green!25}]{R_Y R_Z}
    & \gate[style={fill=purple!12}]{R_X}
    & \gate[style={fill=purple!20}]{R_Y}
    & \gate[style={fill=purple!28}]{R_Z}
    & \gate[style={fill=purple!36}]{R_X}
    & \gate[style={fill=purple!44}]{R_Y}
    & \gate[style={fill=purple!52}]{R_Z}
    &
    & \gate[style={fill=green!32}]{R_Y R_Z}
    & \gate[style={fill=orange!18}]{R_X}
    & \gate[style={fill=orange!26}]{R_Y}
    & \gate[style={fill=orange!34}]{R_Z}
    & \gate[style={fill=orange!42}]{R_X}
    &
    & \gate[style={fill=green!38}]{R_Y R_Z}
    & \gate[style={fill=red!20}]{R_Y}
    &
    & \gate[style={fill=green!45}]{R_Y R_Z}
    & \qw \\
  \lstick{$q_3$}
    & \gate[style={fill=gray!18}]{R_Y}
    &
    & \gate[style={fill=green!18}]{R_Y R_Z}
    & \gate[style={fill=blue!12}]{R_X}
    & \gate[style={fill=blue!22}]{R_Y}
    & \gate[style={fill=blue!32}]{R_Z}
    & \gate[style={fill=blue!42}]{R_X}
    &
    & \gate[style={fill=green!25}]{R_Y R_Z}
    & \gate[style={fill=purple!12}]{R_X}
    & \gate[style={fill=purple!20}]{R_Y}
    & \gate[style={fill=purple!28}]{R_Z}
    & \gate[style={fill=purple!36}]{R_X}
    & \gate[style={fill=purple!44}]{R_Y}
    & \gate[style={fill=purple!52}]{R_Z}
    &
    & \gate[style={fill=green!32}]{R_Y R_Z}
    & \gate[style={fill=orange!18}]{R_X}
    & \gate[style={fill=orange!26}]{R_Y}
    & \gate[style={fill=orange!34}]{R_Z}
    & \gate[style={fill=orange!42}]{R_X}
    &
    & \gate[style={fill=green!38}]{R_Y R_Z}
    & \gate[style={fill=red!20}]{R_Y}
    &
    & \gate[style={fill=green!45}]{R_Y R_Z}
    & \qw
\end{quantikz}%
}
\caption{%
  \textbf{Full Grade-Aware Attention block} at $N{=}4$ qubits (one per constituent),
  showing all active Clifford grades $g\in\{0,1,2,3,4\}$. The mechanism is described in
  Section~\ref{sec:attn_layer}.
  \textcolor{gray!60!black}{\textbf{Gray}} (grade~0, scalar $p{\cdot}P$):
    one $R_Y(\alpha_0)$ per qubit.
  \textcolor{blue!60!black}{\textbf{Blue gradient}} (grade~1, four-momentum):
    $R_X(\alpha^E)\,R_Y(\alpha^{p_x})\,R_Z(\alpha^{p_y})\,R_X(\alpha^{p_z})$.
  \textcolor{purple!60!black}{\textbf{Purple gradient}} (grade~2, bivector $p{\wedge}P$):
    $R_X\,R_Y\,R_Z\,R_X\,R_Y\,R_Z$, one per bivector blade.
  \textcolor{orange!70!black}{\textbf{Orange gradient}} (grade~3, trivector $T_i$):
    $R_X\,R_Y\,R_Z\,R_X$, one per trivector blade.
  \textcolor{red!70!black}{\textbf{Red}} (grade~4, four-volume magnitude $h_i$):
    a single $R_X$.
  The $\mathrm{ZZ}_g$ blocks (multi-qubit) carry the Clifford-pairing coupling angles
  $\eta_g^{(ij)}$, with the symmetry scope specified in Section~\ref{sec:ising_invariance},
  and the green $R_Y R_Z$ are independent variational rotations for each grade and wire.
  Configurations restricted to grades $0$-$3$ simply omit the grade-4 (red) section.}
\label{fig:block_attention}
\end{figure}
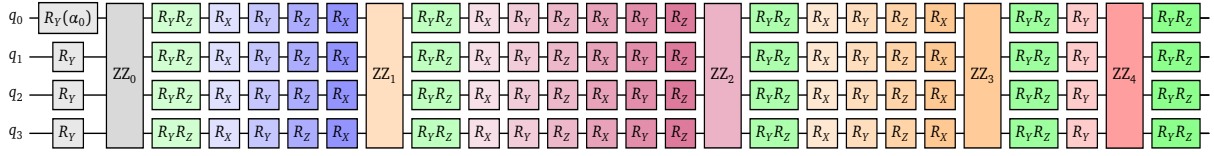

First, in the data-encoding step, the grade-1 angles produced by the encoder
$\mathcal{E}(\mv{x}_i)$ are applied as rotations on qubit $i$, simultaneously for all
$N$ qubits.
Stacking multiple attention layers ($L>1$) re-encodes the same input multivectors at each
layer, which corresponds to an instance of the data re-uploading
technique~\cite{perez_garcia_reuploading}.
In the present work we use a single attention layer ($L=1$), so the grade-1 component is
encoded exactly once per circuit execution.

Second, in the Lorentz-invariant coupling step, for every qubit pair $(i,j)$ with $i<j$ the
two-qubit gate
\begin{equation}
  \mathrm{IsingZZ}(\eta_1^{(ij)})
  = \exp\!\bigl(-i\,\tfrac{\eta_1^{(ij)}}{2}\,Z_i\otimes Z_j\bigr)
  \label{eq:ising_gate}
\end{equation}
is applied, where the angle $\eta_1^{(ij)}$ is the grade-1 Clifford inner product
defined in equation~\eqref{eq:clifford_ip}, which reduces to the Minkowski inner product
$\eta_{\mu\nu}p_i^\mu p_j^\nu$ of the two four-momenta (Section~\ref{sec:minkowski}).
This angle is fixed by the input data and is not a trainable parameter.
By Lemma~\ref{lem:lorentz_invariance}, $\eta_1^{(ij)}$ is a Lorentz scalar, so this
coupling angle is Lorentz invariant for a fixed pair of constituent identities.
The attention block contributes $\binom{N}{2}$ IsingZZ gates.

Third, in the variational-mixing step, rotations $R_Y(\theta_i)\,R_Z(\varphi_i)$ are applied to
each qubit, with parameters $(\theta,\varphi)\in\RR^{N\times2}$ trained via parameter-shift
on hardware (or backpropagation on classical simulators).
These variational parameters control how the encoded grade-1 information is weighted before
the readout stage.

The layer is repeated $L$ times to allow the
circuit to access higher-order polynomial features of the input.
For $G$ active grades the attention stage has $2NGL$ trainable rotation
parameters. These are independent for each layer, grade and data wire and there are
no trainable pairwise attention angles. Encoder scales are shared across tokens and layers.

The state after the attention stage is
\begin{equation}
  \ket{\psi_{\text{attn}}}
  = \prod_{\ell=1}^{L}
    \Bigl[\prod_{g\in\mathcal{G}}^{\leftarrow}
      V_g^{(\ell)}\,\prod_{i<j}\mathrm{IsingZZ}(\eta_g^{(ij)})\,U_{\text{enc},g}\Bigr]
    \ket{0}^{\otimes N},
  \label{eq:attention_state}
\end{equation}
where $\mathcal{G}$ is the ordered active-grade set and the arrow means that
lower grades act first. Each grade executes encoding, pairwise ZZ gates, then
its own $R_Y$ followed by $R_Z$ rotations. Taps follow the final grade of the final layer.

\subsection{Physics-Informed Ancilla Tap}
\label{sec:phase_ancilla}
\label{sec:threebody_formalism}
\label{sec:injected_scalars}

After the attention block, a dedicated ancilla qubit injects a higher-order Lorentz scalar into
the circuit through a Hadamard
test~\cite{cleve1998quantum,nielsen2010quantum,ekert2002direct}. The tap is drawn in
Fig.~\ref{fig:block_tap}.
The injected triple-product coefficients are Lorentz scalars for a fixed constituent set.
This property concerns the coefficients and does not imply invariance of the encoded state
or the complete classifier. Each ancilla is initialised to $\ket{+}=(\ket{0}+\ket{1})/\sqrt{2}$ by a
Hadamard gate, which places it on the equator of the Bloch sphere where it is maximally
sensitive to accumulated phase. It is not measured directly but carried forward, so that its
phase is propagated into the readout stage (Section~\ref{sec:qgatqr_readout}).
The two-particle information is already carried by the grade-aware attention encoding, where
the pairwise Minkowski inner products
\begin{equation}
  \eta_1^{(ij)} \;=\; p_i \cdot p_j \;=\; E_i E_j - \vec{p}_i\cdot\vec{p}_j
  \label{eq:eta1_def}
\end{equation}
set the grade-1 $\mathrm{IsingZZ}$ angles. The tap below adds a three-particle product that this
pairwise encoding does not explicitly inject.

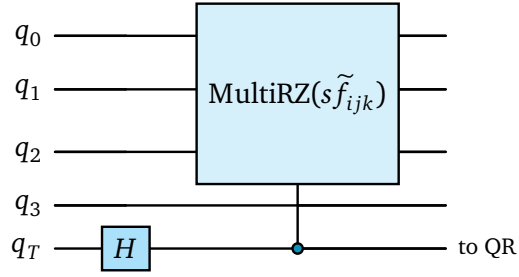
\begin{figure}[H]
\centering
\resizebox{!}{2cm}{%
\begin{quantikz}[column sep=0.55cm, row sep=0.25cm]
  \lstick{$q_0$}
    & \qw
    & \gate[wires=3,style={fill=cyan!14}]{\mathrm{MultiRZ}(s\widetilde f_{ijk})}
    & \qw \\
  \lstick{$q_1$} & \qw & & \qw \\[3pt]
  \lstick{$q_2$} & \qw & & \qw \\
  \lstick{$q_3$} & \qw & & \qw \\
  \lstick{$q_T$}
    & \gate[style={fill=cyan!30}]{H}
    & \ctrl[style={fill=cyan!70!black},color=cyan!70!black]{-2}
    & \qw \rstick{\footnotesize to QR}
\end{quantikz}
}
\caption{%
  \textbf{Block Tap: Hadamard-test ancilla tap} at an illustrative $N{=}4$ data
  qubits. The generator and its gate realisation are given in
  Section~\ref{sec:phase_ancilla}. The ancilla is placed in $|+\rangle$ by a Hadamard,
  controls a diagonal unitary on the data register, and is carried forward to the in-circuit
  readout qubit (``to QR'', Fig.~\ref{fig:block_readout}), never measured directly. The
  \textcolor{cyan!60!black}{\textbf{cyan}} three-particle-product tap $q_T$ applies a
  controlled $\mathrm{MultiRZ}$ over data-qubit triples.}
\label{fig:block_tap}
\end{figure}

A dedicated ancilla qubit $q_T$ records a three-particle Lorentz scalar. For every triple of
constituents $(i,j,k)$ we form the product of the three pairwise Minkowski inner products,
\begin{equation}
  f_{ijk} \;=\; \eta_1^{(ij)}\,\eta_1^{(ik)}\,\eta_1^{(jk)},
  \label{eq:fijk_main}
\end{equation}
using the same grade-1 inner product $\eta_1^{(ij)}$ as the attention layer. The ancilla
$q_T$ controls a $\mathrm{MultiRZ}$ gate, that is a controlled $ZZZ$ rotation, over each
triple with angle $s\,\widetilde f_{ijk}$, where $s$ is a single trainable scale. Because all
$Z_iZ_jZ_k$ generators commute, the controlled block is a single controlled unitary
$U=\exp(-\tfrac{i}{2}\,\hat{G})$ with
\begin{equation}
  \hat{G} \;=\; s\sum_{i<j<k} \widetilde f_{ijk}\,Z_iZ_jZ_k,
  \label{eq:phase_generator}
\end{equation}
and the joint ancilla-data state becomes
$\tfrac{1}{\sqrt{2}}\bigl(\ket{0}\ket{\psi}+\ket{1}\,U\ket{\psi}\bigr)$, so the Hadamard test
writes the accumulated three-particle phase into the coherent state of $q_T$. Each
$\eta_1^{(ij)}$ is a Lorentz scalar (Lemma~\ref{lem:minkowski}), so $f_{ijk}$ and the
generator $\hat{G}$ are unchanged under Lorentz transformations of a fixed constituent set
with fixed qubit assignments. The tap therefore introduces no further reference-vector
restriction on the coupling symmetry described in Section~\ref{sec:ising_invariance}.
The triple-product feature is the product of the three pairwise Minkowski inner products of a
constituent triple. We inject it explicitly as a three-body coupling, which adds this
higher-order term to the quantum attention circuit and improves the reported performance
relative to the no-tap configuration. The choice
is motivated by the top decay $t\to Wb\to q\bar{q}b$, which produces three hard prongs, and by
the standard three-prong shape variables such as the $\tau_{32}$ and energy correlation ratios,
which probe the same kind of structure. The tap registers whether the jet energy is shared among
three centres rather than two, the structure that separates a boosted top jet from the background.
The three-particle-product ancilla adds one qubit, $\binom{N}{3}$ controlled $ZZZ$ rotations,
and a single trainable scale $s$.

Before injection, the triple coefficients are standardized using one mean and
sample standard deviation over all jets and triples in the current batch:
$\widetilde f_{ijk}=(f_{ijk}-\mu_B)/\max(\sigma_B,10^{-6})$.
The attention pairings themselves are not standardized. Thus predictions depend
on the batch composition, including at evaluation time, and no frozen training-set
normalization is used. Before subsequent mixing, the tap's $X$ and $Y$
expectations encode the real and imaginary parts of $\langle\psi|U|\psi\rangle$,
rather than directly measuring each coefficient or $\langle\hat G\rangle$.

\subsection{Variational Layer}
\label{sec:vqc_layer}

The VQC block follows the strongly-entangling / hardware-efficient
ansatz~\cite{schuld_circuit_centric,sim_expressibility} and is drawn in
Fig.~\ref{fig:block_vqc}. For each of $L_{\rm v}$ layers it
applies a trainable single-qubit rotation $R_Y(\phi_{i})\,R_Z(\psi_{i})$ to every qubit,
followed by a ring of nearest-neighbour entangling gates
$\mathrm{CNOT}(q_i,q_{(i+1)\bmod M})$, where $M=N+A$ includes the $A$ enabled tap ancillas
and excludes the readout qubit.
The VQC therefore introduces cross-qubit entanglement, and its rotation
angles are trained jointly with the rest of the circuit.
Ring connectivity can substantially outperform a linear chain in expressibility and
entangling capability~\cite{sim_expressibility}, motivating the ring topology.
The depth $L_{\rm v}$ is kept small (one layer in our experiments) to limit
barren-plateau exposure, consistent with the shallow-ansatz guidance
of~\cite{mcclean_barren}.

The VQC block adds $2M L_{\rm v}$ trainable rotation parameters and $M L_{\rm v}$ two-qubit
gates. It is present in every configuration we report and supplies the model's trainable,
entangling nonlinearity. It acts after the attention and ancilla-tap stages, mixing the data
qubits together with the enabled ancilla before the readout.

\begin{figure}[H]
\centering
\resizebox{!}{2cm}{%
\begin{quantikz}[column sep=0.32cm]
  \lstick{$q_0$}
    & \gate[style={fill=teal!22}]{R_Y R_Z(\phi_0,\psi_0)}
    & \ctrl[style={fill=teal!70!black},color=teal!70!black]{1}
    & \qw & \qw
    & \targ{}
    & \qw \\
  \lstick{$q_1$}
    & \gate[style={fill=teal!22}]{R_Y R_Z(\phi_1,\psi_1)}
    & \targ{}
    & \ctrl[style={fill=teal!70!black},color=teal!70!black]{1}
    & \qw & \qw
    & \qw \\
  \lstick{$q_2$}
    & \gate[style={fill=teal!22}]{R_Y R_Z(\phi_2,\psi_2)}
    & \qw
    & \targ{}
    & \ctrl[style={fill=teal!70!black},color=teal!70!black]{1}
    & \qw
    & \qw \\
  \lstick{$q_3$}
    & \gate[style={fill=teal!22}]{R_Y R_Z(\phi_3,\psi_3)}
    & \qw & \qw
    & \targ{}
    & \ctrl[style={fill=teal!70!black},color=teal!70!black]{-3}
    & \qw
\end{quantikz}
}
\caption{%
  \textbf{Block VQC: VQC activation}, the trainable, entangling activation used in all three
  configurations, shown at an illustrative $N{=}4$ qubits for one layer $L_{\rm v}{=}1$,
  see Section~\ref{sec:vqc_layer}.
  \textcolor{teal!60!black}{\textbf{Teal boxes}}: trainable single-qubit rotations
  $R_Y(\phi_i)R_Z(\psi_i)$ on every qubit.
  \textcolor{teal!70!black}{\textbf{Teal dots}}: a ring of nearest-neighbour
  $\mathrm{CNOT}(q_i,q_{(i+1)\bmod N})$ entangling gates.}
\label{fig:block_vqc}
\end{figure}
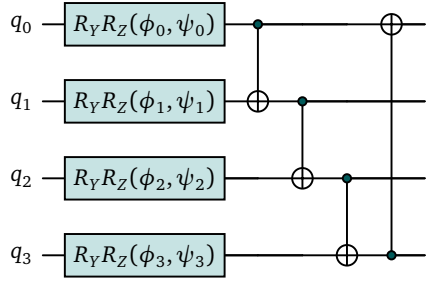

\subsection{Classification Layer of QGrAT}
\label{sec:qgat_readout}

After the variational layer, information about the jet is contained in the joint quantum
state of the data and ancilla qubits. To make a classification decision, we need to turn
measurements of this state into a single score. We first describe a classical combination
of measured correlations, before introducing a readout that collects the signal on an
additional qubit.

\subsubsection{Classical Hamiltonian Observable}
\label{sec:qgatcl_readout}

A direct way to extract information from the prepared state is to measure how the outcomes
of different qubits are related. We use pair correlations because they probe the joint state
built by the attention, ancilla and variational stages. For example, measuring two qubits in
the $Z$ basis gives outcomes labelled $+1$ or $-1$. Their product is positive when the
outcomes agree and negative when they differ. Its average over repeated preparations is the
correlator $\expval{Z_iZ_j}_{\psi}$. Measurements in the $X$ and $Y$ bases probe correlations
along two other qubit measurement axes. Together, these three choices give a broader view
of the state than measurements in a single basis~\cite{nielsen2010quantum}.

For a register of $M=N+A$ qubits, including the $N$ data qubits and $A$ enabled ancillas,
there are $P=M(M-1)/2$ distinct pairs. Collecting one correlator per pair in each basis gives
a vector of $3P$ real-valued features,
\begin{equation}
  \mv{c}
  = \bigl(\expval{Z_iZ_j}_{\psi},\,\expval{X_iX_j}_{\psi},\,\expval{Y_iY_j}_{\psi}\bigr)_{i<j}
  \in\RR^{3P},
  \label{eq:correlators}
\end{equation}
where $\ket{\psi}$ denotes the state after the variational layer. Each entry lies between
$-1$ and $+1$ and summarises one pairwise measurement statistic.

Once these features are available, we learn how strongly each should contribute to the
classification decision. A classical weighted sum, together with a bias $b$, produces the
score, or logit,
\begin{equation}
  \ell = \mv{w}_{\text{ham}}^\top\mv{c}+b,
  \quad \mv{w}_{\text{ham}}\in\RR^{3P}.
  \label{eq:ham_readout}
\end{equation}
The sigmoid $\sigma(\ell)=1/(1+e^{-\ell})$ maps this score to the interval $[0,1]$ for binary
classification. Although the final combination is linear, its inputs are the correlations
produced by the full trainable quantum circuit.

We refer to this construction as a learned Hamiltonian observable because the same weighted
sum can be written as the expectation value of a Hermitian operator~\cite{cerezo_vqe_review}.
By linearity of expectation values,
\begin{equation}
  \ell \;=\; \expval{\hat{H}}_\psi,
  \qquad
  \hat{H} \;=\; \sum_{i<j}\Bigl(w_{ij}^{ZZ}\,Z_iZ_j \,+\, w_{ij}^{XX}\,X_iX_j \,+\, w_{ij}^{YY}\,Y_iY_j\Bigr)+bI.
  \label{eq:ham_observable}
\end{equation}
This is a two-body observable in which each non-identity term acts on two qubits. The identity term
contributes the bias because $\expval{I}_{\psi}=1$. The role of $\hat H$ is to define a
classification observable whose coefficients are learned from labelled jets. Its weights
and bias are classical parameters, trained by backpropagation alongside the circuit
parameters. Their gradients follow directly from the weighted sum and require no
parameter-shift evaluations, since these parameters do not enter any quantum gate.

This choice makes the number of readout parameters grow with the number of qubit pairs.
For $N=10$ and one three-particle-product ancilla, there are $\binom{11}{2}=55$ pairs,
giving $165$ correlator weights and one bias, or $166$ readout parameters. Including the
preceding quantum circuit and encoding, the grades~0--3 model with this readout has $284$
trainable parameters.

On quantum hardware, all $ZZ$ correlators can be estimated from the same set of measurements
of the register in the $Z$ basis, by averaging the products of the outcomes for each pair.
Repeating this procedure in the $X$ and $Y$ bases supplies the remaining features. Thus
three measurement settings suffice, with repeated circuit executions in each setting to
estimate the expectation values. In the exact state-vector simulations used here, these
expectation values are evaluated directly.

\subsubsection{In-Circuit Quantum Readout}
\label{sec:qgatqr_readout}

An alternative to combining measured correlations classically is to collect the
classification signal on an additional qubit before measurement. We use this approach in
the main QGrAT model, adding a readout qubit $q_r$ to the $N$ data qubits and the
three-particle-product ancilla. Its purpose is to turn information distributed across the
register into a signal accessible through a single-qubit observable. We therefore train
its interactions with the other qubits so that its final measurement becomes informative
about the jet label. Figure~\ref{fig:block_readout} follows this process through preparation
of $q_r$, its couplings to the register, local rotations and measurement.

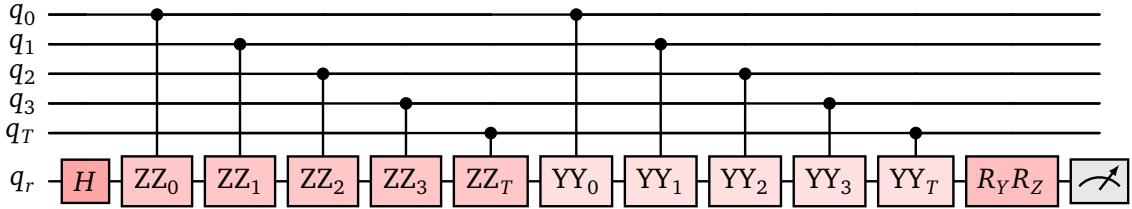
\begin{figure}[H]
\centering
\resizebox{0.95\textwidth}{!}{%
\begin{quantikz}[column sep=0.15cm, row sep=0.2cm]
  \lstick{$q_0$}
    & \qw & \ctrl{5} & \qw & \qw & \qw & \qw
    & \ctrl{5} & \qw & \qw & \qw & \qw
    & \qw & \qw \\
  \lstick{$q_1$}
    & \qw & \qw & \ctrl{4} & \qw & \qw & \qw
    & \qw & \ctrl{4} & \qw & \qw & \qw
    & \qw & \qw \\
  \lstick{$q_2$}
    & \qw & \qw & \qw & \ctrl{3} & \qw & \qw
    & \qw & \qw & \ctrl{3} & \qw & \qw
    & \qw & \qw \\
  \lstick{$q_3$}
    & \qw & \qw & \qw & \qw & \ctrl{2} & \qw
    & \qw & \qw & \qw & \ctrl{2} & \qw
    & \qw & \qw \\
  \lstick{$q_T$}
    & \qw & \qw & \qw & \qw & \qw & \ctrl{1}
    & \qw & \qw & \qw & \qw & \ctrl{1}
    & \qw & \qw \\
  \lstick{$q_r$}
    & \gate[style={fill=red!35}]{H}
    & \gate[style={fill=red!22}]{\mathrm{ZZ}_0}
    & \gate[style={fill=red!22}]{\mathrm{ZZ}_1}
    & \gate[style={fill=red!22}]{\mathrm{ZZ}_2}
    & \gate[style={fill=red!22}]{\mathrm{ZZ}_3}
    & \gate[style={fill=red!22}]{\mathrm{ZZ}_T}
    & \gate[style={fill=red!12}]{\mathrm{YY}_0}
    & \gate[style={fill=red!12}]{\mathrm{YY}_1}
    & \gate[style={fill=red!12}]{\mathrm{YY}_2}
    & \gate[style={fill=red!12}]{\mathrm{YY}_3}
    & \gate[style={fill=red!12}]{\mathrm{YY}_T}
    & \gate[style={fill=red!25}]{R_Y R_Z}
    & \meter[style={fill=gray!18}]{}
\end{quantikz}
}
\caption{%
  \textbf{Block QR: in-circuit Quantum Readout} on an auxiliary qubit $q_r$, drawn
  at an illustrative $N{=}4$ data qubits plus the three-particle-product ancilla
  $q_T$ (any enabled ancilla is coupled the same way), see Section~\ref{sec:qgatqr_readout}.
  \textcolor{red!70!black}{\textbf{Dark red}}: Hadamard places $q_r$ in $|+\rangle$.
  \textcolor{red!60!black}{\textbf{Red blocks}}: individual
  $\mathrm{IsingZZ}(\lambda^Z_i)$ and $\mathrm{IsingYY}(\lambda^Y_i)$ two-qubit gates
  between $q_r$ and each of the data qubits and $q_T$, shown with a control dot on the
  data/ancilla wire and the gate label on $q_r$.
  \textcolor{red!60!black}{\textbf{Light red}}: learned basis rotation
  $R_Y(\mu_0)$ followed by $R_Z(\mu_1)$ on $q_r$.
  \textcolor{gray!50!black}{\textbf{Grey}}: $\langle Z_r\rangle$ measurement,
  feeding $\ell=a\langle Z_r\rangle+b$.}
\label{fig:block_readout}
\end{figure}

We begin with a Hadamard gate, which prepares $q_r$ in the equal superposition
$\ket{+}=(\ket{0}+\ket{1})/\sqrt{2}$. Having both components present allows the subsequent
phase couplings to correlate the readout with the rest of the register. As in the attention
stage, the resulting phase information can then affect measurement probabilities through
later mixing and rotations.

To collect information from all parts of the state, we couple $q_r$ to every data qubit and
to the ancilla $q_T$. First, a row of $\mathrm{IsingZZ}(\lambda^Z_w)$ gates connects each
qubit $w$ to the readout, a second row of $\mathrm{IsingYY}(\lambda^Y_w)$ gates then acts on
the same pairs. Each angle $\lambda^Z_w,\lambda^Y_w\in\RR$ is learned independently, allowing
training to adjust how the readout responds to the state of the other qubits. By this stage,
the ancilla has already interacted with the data through the three-particle tap and the
variational layer. These readout couplings therefore act on the joint state produced by
those stages, carrying information about both the constituents and their correlations into
the final measurement.

After these interactions, we apply $R_Y(\mu_0)$ followed by $R_Z(\mu_1)$ to $q_r$ before
measuring it in the $Z$ basis. The $R_Y$ rotation selects which component of the readout state
is probed by this measurement.\footnote{With this gate order, the final $R_Z$ rotation
commutes with the measured $Z$ observable and leaves its expectation value unchanged.
Its angle $\mu_1$ is retained in the quoted parameter counts, although it is redundant for
this output.}
The expectation value $\expval{Z_r}_{\psi_{\text{full}}}$ is the probability of obtaining
$0$ minus the probability of obtaining $1$, and thus provides a real-valued feature between
$-1$ and $+1$. We convert it into the classification logit through
\begin{equation}
  \ell = a\cdot\expval{Z_r}_{\psi_{\text{full}}} + b,
  \quad a,b\in\RR,
  \label{eq:qr_logit}
\end{equation}
where the classical parameters $a$ and $b$ are learned jointly with the circuit. They set
the scale and offset of the score before the sigmoid maps it to a probability for the binary
cross-entropy loss. In this way, the quantum circuit learns a single readout feature, and
the classical post-processing learns how to map that feature to the classification decision.

With one three-particle-product ancilla and one readout qubit, the complete register contains
$N+2$ qubits. The readout introduces $2(N+1)$ coupling angles, one ZZ and one YY angle for
each data or ancilla qubit, together with the two local rotation angles and the two classical
parameters $a,b$. Its number of coupling parameters therefore grows linearly with the
register size, while the classical observable readout assigns weights to all qubit pairs.
On hardware, this construction uses one measurement setting on a single qubit. Repeated
circuit executions are still needed to estimate its expectation value, whereas in the exact
state-vector simulations used here, that expectation value is evaluated directly.

\section{Experimental Setup}
\label{sec:experimental_setup}

\subsection{Data and Experimental Design}
\label{sec:data}
\label{sec:experiment_design}

We use the Top Tagging dataset~\cite{kasieczka_toptag,toptag_dataset}, a
standard benchmark for jet classification.
Each jet contains up to 200 hadronic constituents with four-momenta $(E,p_x,p_y,p_z)$,
and the task is binary classification of top-quark jets ($t\to Wb\to qqb$) against the QCD
dijet background.
We retain the top-$N$ highest-$p_T$ constituents per jet, normalised by an energy scale
$E_{\text{scale}}=500$\,GeV, and embed each constituent as a grade-1 multivector
$\mv{p}=p^\mu e_\mu$.
Training, validation, and test sets contain 12{,}000, 3{,}000, and 100{,}000 jets respectively,
randomly sampled from a pool of rows in each public split, without explicit class balancing.
We retain the first $N$ constituents in the ordering supplied with the dataset and construct
all active-grade features from their normalised momenta.

Using this common dataset, we study how three architectural choices affect performance.
The first is the ancilla content, either no ancilla or the three-particle-product ancilla. The second is the
grade content of the encoding, either grades $0$ to $3$ or grades $0$ to $4$. The third is the
number of jet constituents $N\in\{8,10\}$, one per data qubit.
Every configuration uses the per-blade grade-aware encoding of the active Clifford grades, the
grade-wise Clifford-coupled $\mathrm{IsingZZ}$ attention, the trainable variational activation, and the
in-circuit $(Z,Y)$ quantum readout.
The training set is fixed at 12K jets and the qubit count is varied to study how performance
scales with circuit width.
We also report one targeted study at a fixed configuration, in which we replace the in-circuit
readout with a classical Hamiltonian-observable head.

To place these comparisons in context, we also train two classical baselines, both built
from the Lorentz-equivariant
\lgatr{} architecture that QGrAT is derived from, and both trained on the same 12K jet split,
the same random seeds, and the same $200$ epochs as the quantum models. The first is
\lgatr{}-Slim, which we deliberately size to QGrAT's parameter budget so the comparison is at
equal capacity and isolates the modelling approach rather than the parameter count. It uses a
single equivariant block with a single attention head, two hidden multivector channels and two
hidden scalar channels, one output scalar channel, and a small multi-layer perceptron head of
hidden width two that maps the pooled Lorentz-invariant features to the logit. This gives $135$
parameters in the equivariant core and $9$ in the head, $144$ in total, matching the QGrAT model
at $N{=}10$. The second is the full \lgatr{} at its minimal single-block configuration with one
hidden multivector channel and one hidden scalar channel, which comes to $1{,}215$ parameters
and serves as a classical capacity reference rather than a parameter-matched comparison.
All quantum models are trained with early stopping, and results are reported as multi-seed mean
and standard deviation.
The complete set of models, qubit counts, and test metrics is given in
Table~\ref{tab:results}.

\subsection{Training and Optimisation}
\label{sec:training}

For the model comparisons described above, we use a common optimisation protocol.
Models are trained with the Adam optimiser using separate learning rates for quantum and
classical parameters. Quantum parameters use $\eta_{\rm Q}=10^{-3}$ (parameter-shift on
hardware, backpropagation on the classical simulator) and classical parameters use $\eta_{\rm C}=5\times10^{-3}$.
A cosine annealing schedule uses a period of $2T$ epochs and a shared minimum
learning rate of $5\times10^{-5}$, where $T$ is the maximum epoch count.
The batch size is $32$. Gradient norms are clipped separately at $10.0$ for the
quantum group and $1.0$ for the classical group. Early stopping uses a minimum
validation-AUC improvement of $0.0005$ over $25$ epochs.
Attention and readout rotations belong to the quantum optimiser group.
Encoder gains, tap scales, VQC parameters and output/head parameters belong
to the classical group, although the VQC parameters also control quantum gates.
The maximum epoch count is $200$ for both the QGrAT variants and the \lgatr{} baselines. With
the early-stopping rule above, most runs end before this limit, so the per-seed training
histories vary in length.

Quantum simulation uses PennyLane for exact state-vector simulation, with gradients
obtained by backpropagation.
All experiments in the present version of this work are conducted on classical simulators, and
hardware experiments on IBM Heron are deferred to future work.

\subsection{Evaluation and Diagnostics}
\label{sec:metrics}
\label{sec:uncertainty}
\label{sec:diagnostics}

To assess classification performance, we report the binary cross-entropy loss and the
area under the receiver operating characteristic curve (AUC). We also report classification
accuracy, precision, recall and the F1 score at the standard threshold $0.5$.
For comparison with the high-energy physics literature, we additionally report the
background rejection $1/\epsilon_B$ at signal efficiencies $\epsilon_S\in\{0.30,0.50\}$,
which is a standard metric for jet taggers~\cite{kasieczka_toptag}.
Published classical baselines on this benchmark achieve
$1/\epsilon_B|_{\epsilon_S=0.30}\approx 1300\text{--}2500$~\cite{lorentznet,pelican,lgatr}.

To judge the stability of these performance estimates, we distinguish two sources of
uncertainty. The first is the spread across random training seeds,
which we report as the standard deviation over seeds. The second comes from the finite size of
the test set. The metrics are measured on $100{,}000$ jets, and a different draw of jets would
give slightly different numbers. We estimate this second term by resampling the saved per-jet
scores with replacement and recomputing each metric many times, then taking the central $95\%$
of the resampled values as an interval.

Alongside the final classification metrics, we examine how the models learn during
training. For each epoch, we record the learning rate, gradient $\ell_2$-norm and pooled
gradient-element variance separately for the quantum and classical parameter groups,
$\theta_{\rm Q}$ and $\theta_{\rm C}$. We also track the change in validation AUC over the
most recent ten epochs. Together, these quantities describe the optimisation trajectory
and the gradient signals available to update the model.

For a parameter group $\theta$ with loss $L$, we quantify the magnitude of this update
signal through the gradient $\ell_2$-norm,
\begin{equation}
\lVert\nabla_\theta L\rVert_2 = \big(\sum_k (\partial L/\partial\theta_k)^2\big)^{1/2}.
\end{equation}
It measures the magnitude of the update signal available to the optimiser. A norm that
collapses towards zero indicates a flat loss landscape in which no parameter receives a
usable gradient, whereas a stable, non-vanishing norm indicates that the optimiser can make
progress. We report the norm separately for the quantum and classical groups because the
two are trained with different learning rates and can exhibit different landscapes.
The reported norm is the post-clipping group norm from the final batch of the epoch.

To describe the variation in these update signals, we compute the sample variance of the
gradient elements in each parameter group after
clipping, pooling all parameter coordinates and all mini-batches within an epoch.
This statistic combines variation across parameter coordinates and training batches.
As a reference we use the $2^{-N}$ scale
from the quantum-circuit trainability literature~\cite{mcclean_barren}, at which the gradient
variance of generic deep circuits becomes exponentially small in the qubit count. We track the
gradient norm and this variance during training and report where they stand relative to the
$2^{-N}$ scale over the simulated range.
The measured training dynamics and gradient diagnostics are presented in
Section~\ref{sec:results_training}.

\section{Results}
\label{sec:results}

\subsection{Test Metrics}
\label{sec:results_main}

Table~\ref{tab:results} reports test metrics for the \qgat{} model family and the classical
\lgatr{} baseline on the Top Tagging benchmark (12K training jets, 100k test jets, mean\,$\pm$\,std over
random seeds). The grade-aware attention processes all five Clifford grades ($0$-$4$) of each
token. The table spans the full ladder over the ancilla content (no ancilla and
three-particle-product ancilla), grade content (grades $0$-$3$ and $0$-$4$),
and qubit count ($N{=}8,10$), all with the variational block and the in-circuit quantum
readout. The three-particle-product ancilla at grades $0$-$3$ gives the strongest rejection,
the grade-4 magnitude feature ($0$-$4$) gives slightly lower rejection, as detailed in the comparison of model variants in
Section~\ref{sec:systematic_experiments}.

\begin{table}[t]
\centering
\caption{%
  Test metrics on the Top Tagging benchmark (12K training jets, 100k test jets). Values are
  mean\,$\pm$\,std over random seeds. The two constituent counts $N{=}8$ and $N{=}10$ are
  shown side by side for each configuration.
  AUC: test ROC-AUC. $1/\epsilon_B^{50}$, $1/\epsilon_B^{30}$: background rejection at
  signal efficiency $\varepsilon_S{=}0.50$ and $0.30$ (higher is better).
  \textbf{Bold}: best \qgat{} configuration.
  Each $N$ block carries its own trainable parameter count $P$. The \qgat{} counts grow
  with $N$ because one data qubit is added per constituent, whereas the two \lgatr{} baselines
  are $N$-independent. At $N{=}10$ the strongest \qgat{} configuration and the slim \lgatr{}
  baseline are compared at an identical budget of $144$ parameters.
  All \qgat{} rows use grade-aware encoding of the active Clifford grades per blade, a
  variational block, and the in-circuit quantum readout, and differ only in ancilla content
  and grade content.
}
\label{tab:results}
\renewcommand{\arraystretch}{1.3}
\setlength{\tabcolsep}{5pt}
\resizebox{\textwidth}{!}{%
\begin{tabular}{l rrrr rrrr}
\toprule
 & \multicolumn{4}{c}{$N=8$} & \multicolumn{4}{c}{$N=10$} \\
\cmidrule(lr){2-5}\cmidrule(lr){6-9}
\textbf{Grades} & $P$ & \textbf{AUC} & $1/\epsilon_B^{50}$ & $1/\epsilon_B^{30}$
                & $P$ & \textbf{AUC} & $1/\epsilon_B^{50}$ & $1/\epsilon_B^{30}$ \\
\midrule
\multicolumn{9}{@{}l}{\emph{No ancilla}}\\
\quad grades 0-3 & $115$ & $0.9230{\pm}0.0005$ & $18.8{\pm}0.2$ & $36.2{\pm}0.6$
                  & $139$ & $0.9296{\pm}0.0006$ & $19.9{\pm}0.3$ & $36.4{\pm}0.6$ \\
\quad grades 0-4 & $132$ & $0.9204{\pm}0.0021$ & $18.0{\pm}0.8$ & $33.4{\pm}2.4$
                  & $160$ & $0.9294{\pm}0.0010$ & $19.8{\pm}0.3$ & $36.5{\pm}1.1$ \\
\addlinespace[2pt]
\multicolumn{9}{@{}l}{\emph{Three-particle-product ancilla}}\\
\quad grades 0-3 & $120$ & $0.9353{\pm}0.0013$ & $30.0{\pm}1.2$ & $90.4{\pm}8.5$
                  & $144$ & $\mathbf{0.9439{\pm}0.0017}$ & $\mathbf{35.7{\pm}2.4}$ & $\mathbf{96.8{\pm}7.7}$ \\
\quad grades 0-4 & $137$ & $0.9356{\pm}0.0003$ & $30.2{\pm}0.5$ & $90.4{\pm}9.7$
                  & $165$ & $0.9424{\pm}0.0012$ & $33.5{\pm}1.6$ & $89.6{\pm}5.7$ \\
\midrule
\multicolumn{9}{@{}l}{\emph{Classical baseline}}\\
\quad \lgatr{}-Slim & $144$ & $0.9403{\pm}0.0037$ & $36.4{\pm}6.4$ & $109.7{\pm}36.6$
                  & $144$ & $0.9404{\pm}0.0098$ & $30.4{\pm}13.9$ & $75.1{\pm}53.8$ \\
\quad \lgatr{} & $1{,}215$ & $0.9465{\pm}0.0009$ & $48.2{\pm}4.4$ & $149.9{\pm}15.3$
                  & $1{,}215$ & $0.9579{\pm}0.0012$ & $72.7{\pm}3.3$ & $192.0{\pm}9.6$ \\
\bottomrule
\end{tabular}%
}
\end{table}

At $N{=}10$ the strongest configuration is the three-particle-product ancilla at grades
$0$-$3$, which reaches test AUC $0.9439{\pm}0.0017$ with $144$ trainable parameters and
background rejection $1/\varepsilon_B^{30}=96.8$ and
$1/\varepsilon_B^{50}=35.7$.
Figure~\ref{fig:ladder_metrics} summarises the test metrics across the full configuration
ladder at $N{=}10$.
The per-axis controlled comparison, how the ancilla content, the grade content, and the
constituent count each move these metrics, is examined in
Section~\ref{sec:systematic_experiments}.
Table~\ref{tab:results} includes two classical \lgatr{} baselines for reference. At an exactly
matched parameter budget the slim \lgatr{} ($144$ parameters) reaches AUC $0.9404{\pm}0.0098$ at
$N{=}10$, which sits just below the strongest QGrAT configuration ($0.9439{\pm}0.0017$,
$144$ parameters). At equal size QGrAT is therefore at least as accurate on AUC and ahead on
both rejection metrics, reaching $1/\varepsilon_B^{50}=35.7$ against $30.4$ and
$1/\varepsilon_B^{30}=96.8$ against $75.1$. The full \lgatr{} at its minimal single-block
configuration ($1{,}215$ parameters) is clearly ahead, reaching AUC $0.9579{\pm}0.0012$ with
roughly $8\times$ the parameters, which sets the classical capacity reference rather than a
matched comparison.

A clear difference between QGrAT and the matched slim \lgatr{} is the stability of the
results across seeds. At $N{=}10$ QGrAT holds its background rejection to
$1/\varepsilon_B^{50}=35.7{\pm}2.4$ and $1/\varepsilon_B^{30}=96.8{\pm}7.7$, whereas the
slim \lgatr{} swings much more widely at $30.4{\pm}13.9$ and $75.1{\pm}53.8$. The AUC is
similarly steadier for QGrAT ($\pm0.0017$ versus $\pm0.0098$). We read this as a regularising
effect of the symmetry-informed circuit, whose attention angles are fixed Clifford
inner products rather than free weights, leaving only a small variational block and the
ancilla gains to train. At a matched parameter count QGrAT is therefore both slightly more
accurate than the slim \lgatr{} and markedly more reproducible on the rejection metrics that
matter most for tagging.

At $100{,}000$ test jets the finite-sample term is small. The test set is class balanced, so
about $50{,}000$ jets are background, and several hundred survive even the tightest cut reported
here. For the strongest model at $N{=}10$ the finite-sample $95\%$ interval is about $\pm0.001$
on the AUC, $\pm2$ on the rejection at $50\%$ efficiency, and $\pm8$ on the rejection at $30\%$
efficiency, comparable to the seed spread. Table~\ref{tab:uncertainty} reports both terms side
by side for the headline configurations. For the matched classical \lgatr{}-Slim the across-seed
spread in the rejection metrics is far larger than the finite-sample interval.

\begin{table}[H]
\centering
\caption{%
  Finite-sample uncertainty against seed scatter at $N{=}10$ on the $100{,}000$-jet test set.
  For each headline configuration the table gives the seed mean with its across-seed standard
  deviation (mean${\pm}$sd) and, separately, the median finite-sample $95\%$ bootstrap
  half-width ($\mathrm{CI_{fs}}$). The two rejections are $1/\epsilon_B$ at signal efficiencies
  $0.50$ and $0.30$.
}
\label{tab:uncertainty}
\renewcommand{\arraystretch}{1.15}
\setlength{\tabcolsep}{5pt}
\resizebox{0.85\columnwidth}{!}{%
\begin{tabular}{l rr rr rr}
\toprule
 & \multicolumn{2}{c}{\textbf{AUC}} & \multicolumn{2}{c}{$1/\epsilon_B^{50}$}
 & \multicolumn{2}{c}{$1/\epsilon_B^{30}$}\\
\cmidrule(lr){2-3}\cmidrule(lr){4-5}\cmidrule(lr){6-7}
\textbf{Configuration} & mean${\pm}$sd & $\mathrm{CI_{fs}}$ & mean${\pm}$sd & $\mathrm{CI_{fs}}$
                       & mean${\pm}$sd & $\mathrm{CI_{fs}}$\\
\midrule
Three-particle, grades 0-3 & $0.9439{\pm}0.0017$ & ${\pm}0.0014$ & $35.7{\pm}2.4$ & ${\pm}1.9$ & $96.8{\pm}7.7$ & ${\pm}8.3$\\
Three-particle, grades 0-4 & $0.9424{\pm}0.0012$ & ${\pm}0.0014$ & $33.5{\pm}1.6$ & ${\pm}1.8$ & $89.6{\pm}5.7$ & ${\pm}7.9$\\
No ancilla, grades 0-3     & $0.9296{\pm}0.0006$ & ${\pm}0.0017$ & $19.9{\pm}0.3$ & ${\pm}0.8$ & $36.4{\pm}0.6$ & ${\pm}2.0$\\
\lgatr{}-Slim ($144$)      & $0.9404{\pm}0.0098$ & ${\pm}0.0015$ & $30.4{\pm}13.9$ & ${\pm}1.1$ & $75.1{\pm}53.8$ & ${\pm}3.3$\\
\lgatr{} ($1{,}215$)       & $0.9579{\pm}0.0012$ & ${\pm}0.0011$ & $72.7{\pm}3.3$ & ${\pm}5.7$ & $192.0{\pm}9.6$ & ${\pm}25.1$\\
\bottomrule
\end{tabular}%
}
\end{table}

A further study at $N{=}10$ asks whether the readout has to be quantum at all.
Table~\ref{tab:hobs} replaces the in-circuit quantum readout with a classical
Hamiltonian head that learns a linear function of the pairwise Pauli-correlator feature vector of
the same variational circuit. The quantum attention circuit and the variational block are unchanged
and only the readout is classical. For the three-particle-product ancilla the multi-seed
classical runs are complete. At grades 0-3 the classical head is slightly ahead of the
quantum readout on all three metrics, and at grades 0-4 the two are comparable within seed
scatter. Either way the classical and quantum readouts reach the same level of performance,
which indicates that the discriminating information is carried by the quantum attention
circuit and can be extracted by either head.

\begin{table}[H]
\centering
\caption{%
  Readout-layer study at $N{=}10$. The in-circuit quantum readout versus a classical
  Hamiltonian head that learns a linear function of the pairwise Pauli-correlator feature
  vector of the same variational circuit. The quantum attention circuit and variational block are
  identical across the two and only the extraction differs. The baseline columns repeat the
  $N{=}10$ quantum-readout values of Table~\ref{tab:results}. Values are mean\,$\pm$\,std over
  random seeds.
}
\label{tab:hobs}
\renewcommand{\arraystretch}{1.15}
\setlength{\tabcolsep}{5pt}
\resizebox{0.85\textwidth}{!}{%
\begin{tabular}{l rrr rrr}
\toprule
 & \multicolumn{3}{c}{Quantum readout} & \multicolumn{3}{c}{Classical Hamiltonian readout} \\
\cmidrule(lr){2-4}\cmidrule(lr){5-7}
\textbf{Configuration} & \textbf{AUC} & $1/\epsilon_B^{50}$ & $1/\epsilon_B^{30}$
                       & \textbf{AUC} & $1/\epsilon_B^{50}$ & $1/\epsilon_B^{30}$ \\
\midrule
\multicolumn{7}{@{}l}{\emph{No ancilla}}\\
\quad grades 0-3 & $0.9296{\pm}0.0006$ & $19.9{\pm}0.3$ & $36.4{\pm}0.6$
                  & $0.9269{\pm}0.0019$ & $18.5{\pm}0.8$ & $32.3{\pm}2.6$ \\
\quad grades 0-4 & $0.9294{\pm}0.0010$ & $19.8{\pm}0.3$ & $36.5{\pm}1.1$
                  & $0.9286{\pm}0.0016$ & $19.1{\pm}0.8$ & $35.3{\pm}1.7$ \\
\addlinespace[2pt]
\multicolumn{7}{@{}l}{\emph{Three-particle-product ancilla}}\\
\quad grades 0-3 & $0.9439{\pm}0.0017$ & $35.7{\pm}2.4$ & $96.8{\pm}7.7$
                  & $0.9475{\pm}0.0014$ & $42.7{\pm}2.7$ & $109.4{\pm}8.7$ \\
\quad grades 0-4 & $0.9424{\pm}0.0012$ & $33.5{\pm}1.6$ & $89.6{\pm}5.7$
                  & $0.9449{\pm}0.0016$ & $35.9{\pm}3.3$ & $92.7{\pm}8.9$ \\
\bottomrule
\end{tabular}%
}
\end{table}

\begin{figure}[t]
\centering
\includegraphics[width=0.95\linewidth]{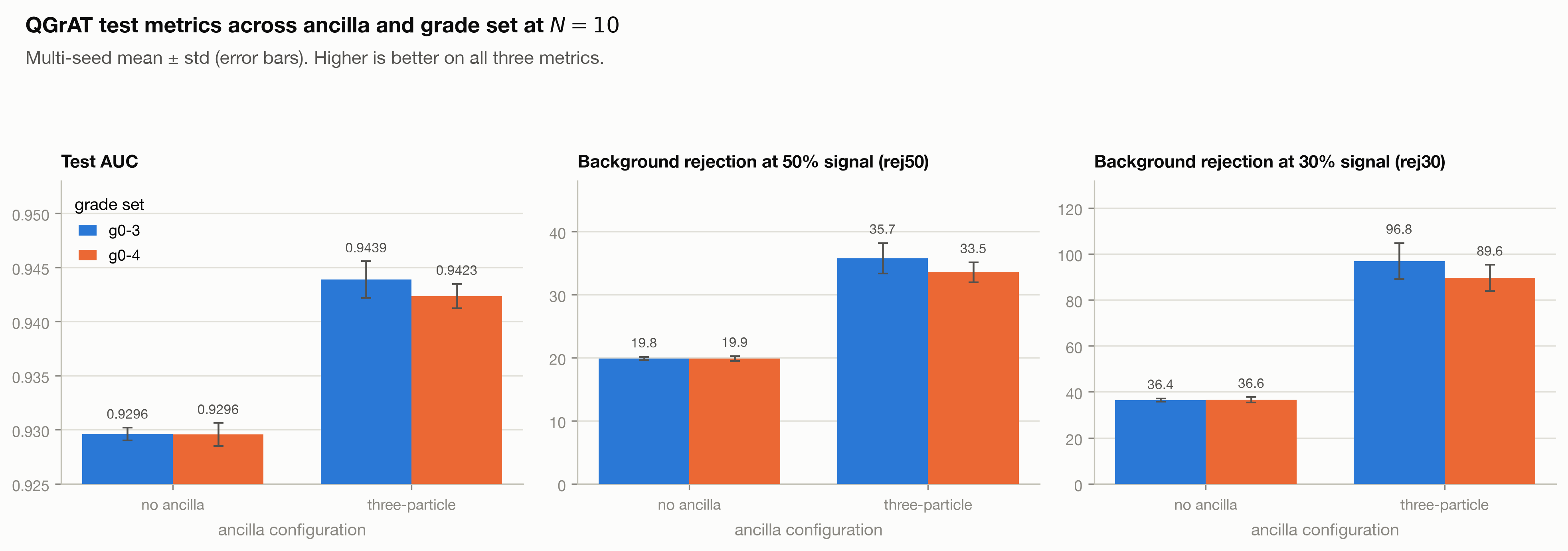}\\[0.8ex]
\includegraphics[width=0.95\linewidth]{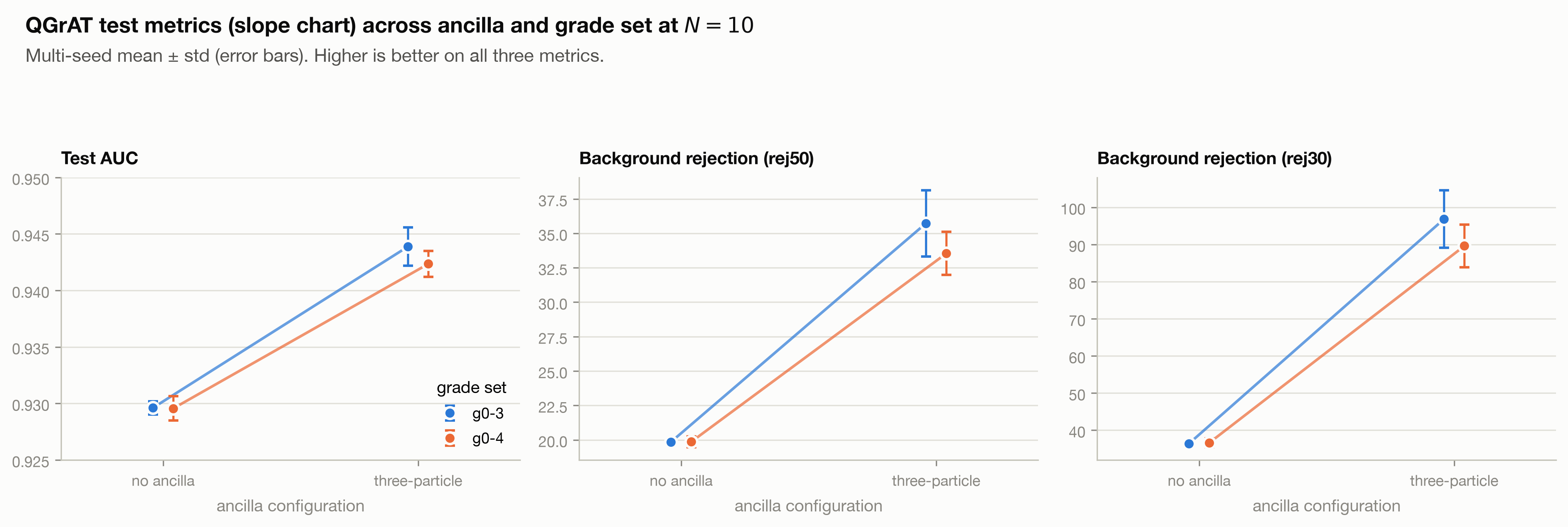}
\caption{Ancilla configuration versus grade set at $N=10$, 12K training jets, 100k test jets, the same
multi-seed data shown two ways. Top: grouped bars of test AUC and background rejection at
signal efficiencies $0.50$ and $0.30$, mean with error bars for the standard deviation, for the
ancilla configurations (no ancilla and three-particle-product) at grades $0$-$3$ and $0$-$4$. Bottom: the same values as grade-set
slopes across the ancilla configurations.}
\label{fig:ladder_metrics}
\end{figure}

\subsection{Training}
\label{sec:results_training}

Figure~\ref{fig:training_curves} shows the validation AUC and BCE loss as a function of
training epoch at $N=10$, for the three ancilla configurations at grades $0$-$3$ and $0$-$4$.
\qgat{} begins near random performance at the first epoch, then converges over the first
${\sim}25$ epochs as the three-particle-product ancilla and readout qubit align. The initial
near-random performance reflects the Hadamard initialisation of the ancilla and readout
qubits, as discussed in Sections~\ref{sec:phase_ancilla} and~\ref{sec:qgatqr_readout}.
The effect of the ancilla is concentrated in the background rejection. The same frontend and
readout without it reach AUC ${\sim}0.930$, whereas the three-particle-product ancilla raises
the AUC to $0.944$ and nearly doubles the high-purity rejection (Table~\ref{tab:results}).

We next examine the gradient diagnostics defined in Section~\ref{sec:diagnostics}.
At $N=10$ the quantum gradient variance stays above the $2^{-10}$ reference
scale\footnote{At $N=10$ the total circuit width is $N+2=12$ qubits including the tap ancilla and readout qubit.
We quote the scale in the conventional form $2^{-N}$ in the data-qubit count $N$.} throughout
training and does not collapse as training proceeds.
Across the qubit counts we simulate (Figure~\ref{fig:scaling_diagnostics}), the
gradient variance stays above the $2^{-N}$ scale at every count tested.
The main qualitative change with increasing $N$ is a slower convergence, as the in-circuit
readout must jointly optimise with a progressively larger data register.

\begin{figure}[tp]
\centering
\includegraphics[width=0.8\linewidth]{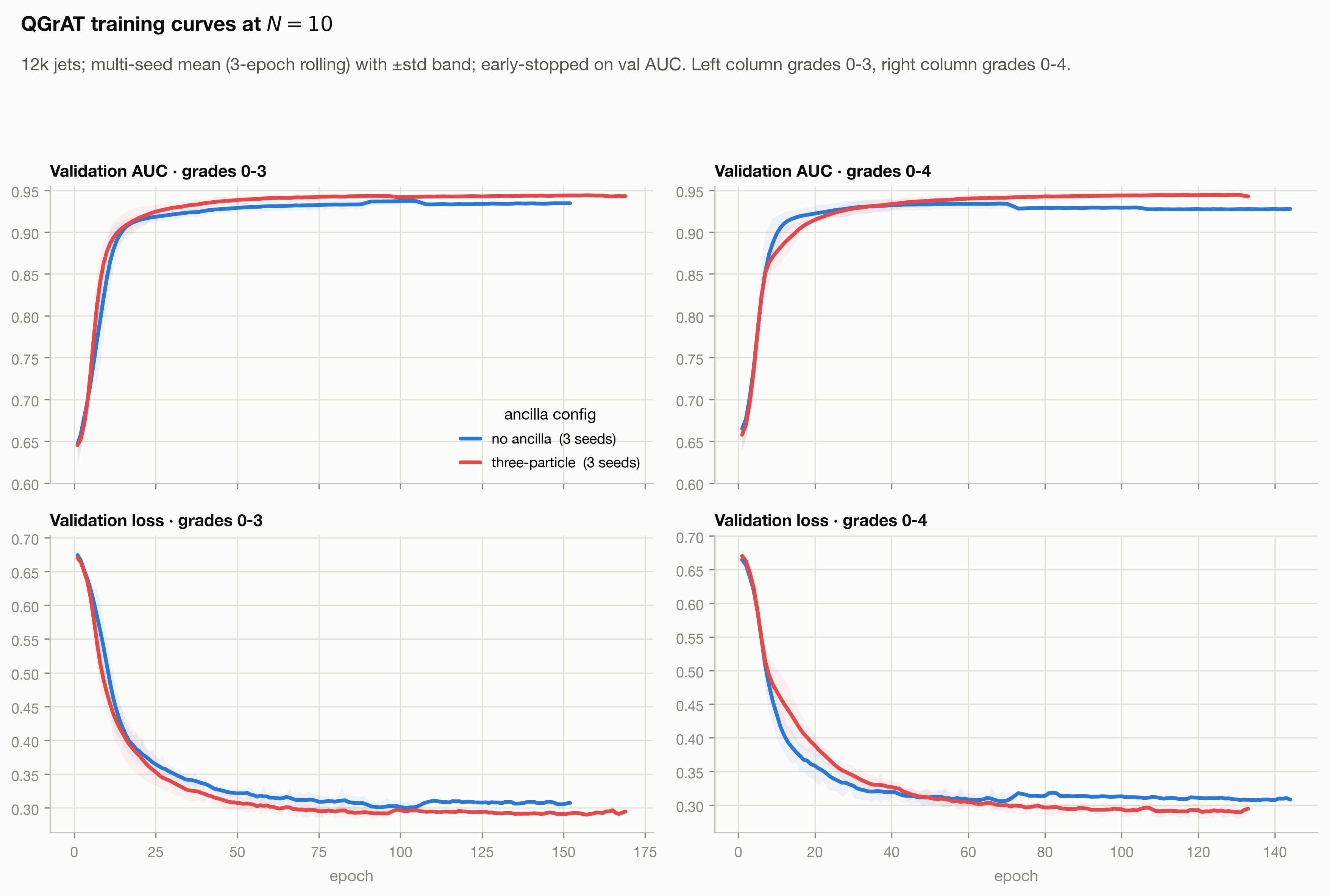}
\caption{Validation AUC (left) and validation loss (right) versus training epoch at $N=10$,
12K training jets, for representative QGrAT configurations (no ancilla and three-particle-product
ancilla, at grades $0$-$3$ and $0$-$4$). Curves
show the multi-seed mean with a shaded band for the standard deviation. Convergence occurs
within roughly $25$ epochs once the ancilla and readout qubit align.}
\label{fig:training_curves}

\vspace{1.5ex}
\includegraphics[width=0.8\linewidth]{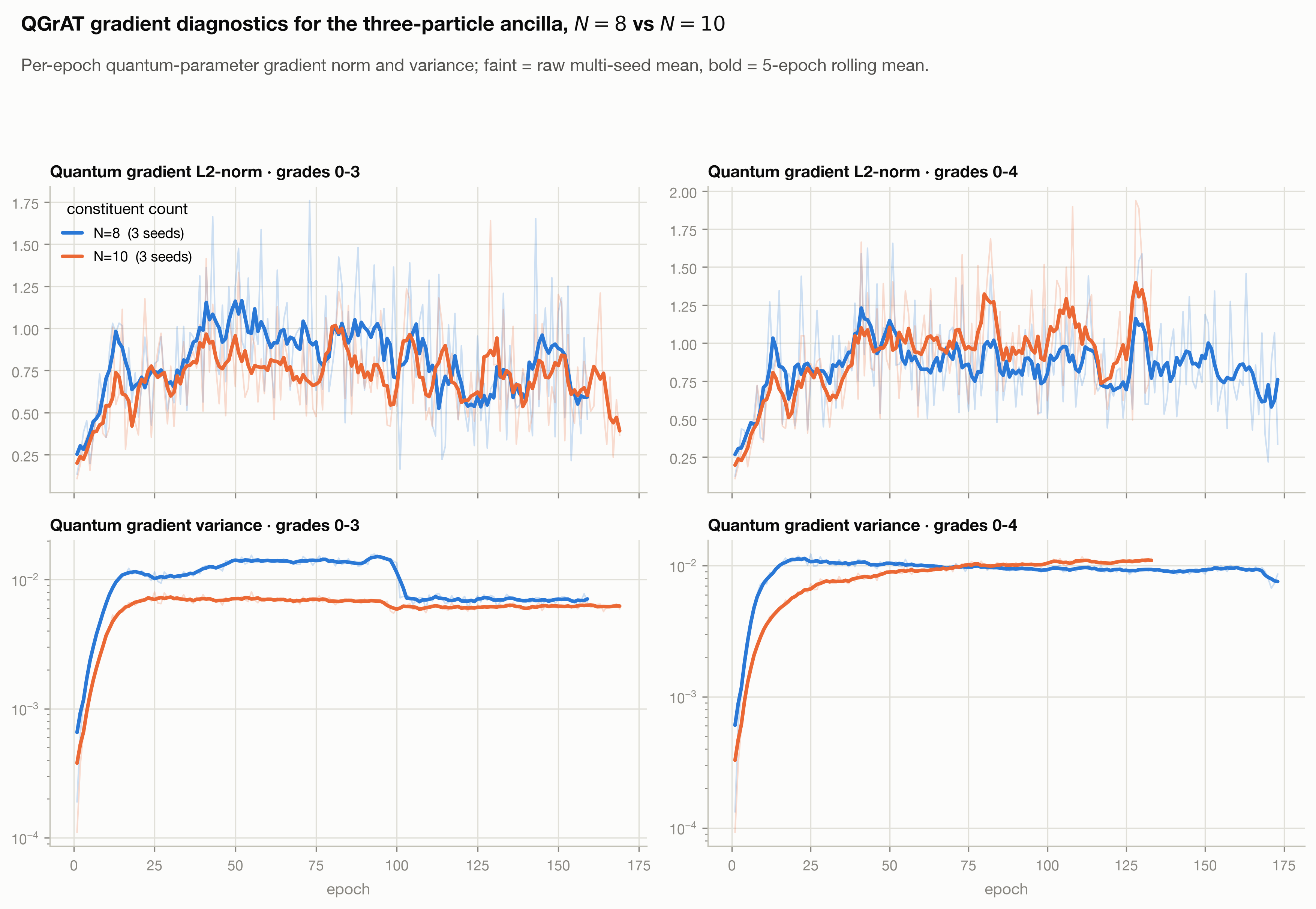}
\caption{Quantum gradient diagnostics versus training epoch for the three-particle-product
ancilla at grades $0$-$3$, comparing $N=8$ and $N=10$. The left panel shows the quantum
gradient $\ell_2$-norm and the right panel the quantum gradient variance on a logarithmic scale,
plotted as multi-seed mean with a shaded band for the standard deviation. The gradient variance
stays above the $2^{-N}$ reference~\cite{mcclean_barren} at both qubit counts over the
simulated range. These observations do not claim the absence of barren plateaus or
favourable asymptotic gradient scaling.}
\label{fig:scaling_diagnostics}
\end{figure}

\section{Systematic Comparison of Model Variants}
\label{sec:systematic_experiments}

Having established the overall tagging performance, we now examine which modelling choices
are most useful. We first study the three-particle-product ancilla and test the role of its
coupling. We then vary the geometric features presented to the circuit and the number of
retained constituents. These comparisons follow the experimental design of
Section~\ref{sec:experiment_design}: models share the grade-aware encoding, variational block
and in-circuit quantum readout, while we vary the selected component and hold the other
configuration choices fixed. The separate comparison of quantum and classical readouts is
reported in Section~\ref{sec:results}. Throughout this section, we interpret the mean test
metrics together with their variation across training seeds.

\subsection{Ancilla and Three-Particle Coupling}
\label{sec:exp_ancilla}
\label{sec:exp_controls}

We begin by asking whether the circuit benefits from an explicit interaction involving
three constituents, once pairwise attention and variational mixing are already present.
At $N{=}10$ with grades $0$ to $3$, adding the three-particle-product ancilla raises the AUC
from $0.9296$ to $0.9439$ (Table~\ref{tab:results} and Fig.~\ref{fig:ladder_metrics}). The
change is especially pronounced in background rejection, rising from $19.9$ to $35.7$
at $50\%$ signal efficiency and from $36.4$ to $96.8$ at $30\%$ signal efficiency. Thus,
for the same fraction of retained top jets, substantially fewer background jets pass the
selection. This is why we examine fixed-efficiency rejection alongside AUC, which summarises
performance over the full range of classification thresholds.

Adding an ancilla also enlarges the register and gives the variational and readout layers
another qubit to interact with. This raises a more targeted question about how much the
explicit three-particle coupling matters when the ancillary qubit is already present.
To address this, we retain the ancilla and its surrounding circuit but switch off its
triple-product coupling. We use grade~1 alone for this control, retaining the pairwise
Minkowski attention in both cases.

For this grade-1 model, disabling the coupling lowers the AUC from $0.9007$ to $0.7337$ and
the background rejection at $50\%$ signal efficiency from $26.9$ to $4.4$
(Table~\ref{tab:grade_ladder}). Retaining the extra qubit and its variational and readout
connections is therefore insufficient to reproduce the performance of the coupled model.
This control demonstrates the importance of the three-particle interaction in the grade-1
configuration. Together with the gain seen for grades $0$ to $3$, it motivates examining
how the ancilla works with the different geometric features. The control does not separately
quantify the contribution of the coupling for every grade combination.

\subsection{Grade Content and Constituent Count}
\label{sec:exp_fullgrade}
\label{sec:exp_scaling}

We can change the information presented to the circuit in two complementary ways, either by
constructing different geometric features from the retained momenta or by retaining more
constituents of the jet. To study the first choice, we initially fix $N{=}10$ and compare
progressively larger sets of Clifford grades. Table~\ref{tab:grade_ladder} gives the results
with and without the three-particle-product ancilla, including the grade-1 control discussed
above. Comparing the two sets of models also lets us ask whether the usefulness of a grade
depends on the presence of the ancilla.

\begin{table}[t]
\centering
\caption{%
  Grade-content ladder at $N{=}10$, 12K training jets, 100k test jets, built one Clifford grade at a time,
  without an ancilla and with the three-particle-product ancilla. All rows use the grade-aware
  encoding of the listed grades, the variational block and the in-circuit quantum readout, and
  differ only in the active grades and the ancilla. The ``grade 1, no three-particle coupling''
  row is the control that keeps the grade-1 ancilla model but switches off only the
  three-particle coupling, leaving the pairwise attention couplings in place.
  Values are mean\,$\pm$\,std over three or more
  random seeds. AUC: test ROC-AUC. $1/\epsilon_B^{50}$,
  $1/\epsilon_B^{30}$: background rejection at signal efficiency $0.50$ and $0.30$.
}
\label{tab:grade_ladder}
\renewcommand{\arraystretch}{1.15}
\setlength{\tabcolsep}{5pt}
\begin{tabular}{l rrr}
\toprule
\textbf{Grades} & \textbf{AUC} & $1/\epsilon_B^{50}$ & $1/\epsilon_B^{30}$ \\
\midrule
\multicolumn{4}{@{}l}{\emph{No ancilla}}\\
\quad grade 0    & $0.9226{\pm}0.0019$ & $18.7{\pm}0.9$ & $34.6{\pm}2.9$ \\
\quad grades 0-1 & $0.9297{\pm}0.0014$ & $19.9{\pm}0.7$ & $37.2{\pm}2.4$ \\
\quad grades 0-2 & $0.9299{\pm}0.0002$ & $20.2{\pm}0.1$ & $37.4{\pm}0.8$ \\
\quad grades 0-3 & $0.9296{\pm}0.0006$ & $19.9{\pm}0.3$ & $36.4{\pm}0.6$ \\
\quad grades 0-4 & $0.9294{\pm}0.0010$ & $19.8{\pm}0.3$ & $36.5{\pm}1.1$ \\
\addlinespace[2pt]
\multicolumn{4}{@{}l}{\emph{Three-particle-product ancilla}}\\
\quad grade 1    & $0.9007{\pm}0.0031$ & $26.9{\pm}0.7$ & $54.1{\pm}0.9$ \\
\quad grade 1, no three-particle coupling & $0.7337{\pm}0.0022$ & $4.4{\pm}0.1$ & $8.0{\pm}0.1$ \\
\quad grades 0-1 & $0.9434{\pm}0.0021$ & $34.1{\pm}2.7$ & $94.2{\pm}14.1$ \\
\quad grades 0-2 & $0.9421{\pm}0.0005$ & $32.9{\pm}1.5$ & $88.2{\pm}7.6$ \\
\quad grades 0-3 & $0.9439{\pm}0.0017$ & $35.7{\pm}2.4$ & $96.8{\pm}7.7$ \\
\quad grades 0-4 & $0.9424{\pm}0.0012$ & $33.5{\pm}1.6$ & $89.6{\pm}5.7$ \\
\bottomrule
\end{tabular}
\end{table}

Without an ancilla, much of the performance is already obtained from the scalar feature
$p_i\cdot P$, which relates each constituent to the total jet momentum. Using this feature
alone gives an AUC of $0.9226$ and background rejection of $18.7$ at $50\%$ signal efficiency.
Adding grade~1 raises these values to $0.9297$ and $19.9$. Including the bivector, trivector
and four-volume-magnitude features then produces only small changes, with the rejection
remaining close to $20$. In this configuration, the richer geometric representation does
not translate into a clear additional gain beyond the scalar and vector features.

With the ancilla present, the scalar feature again makes a substantial contribution.
Adding it to the grade-1 model raises the AUC from $0.9007$ to $0.9434$ and the rejection at
$50\%$ signal efficiency from $26.9$ to $34.1$. Further grades give a less uniform pattern, with
the mean rejection falling to $32.9$ for grades $0$ to $2$, then rising to $35.7$ for grades
$0$ to $3$. The latter configuration has the highest mean rejection in the ladder, but its
improvement over grades $0$ to $1$ is small compared with the reported variation across
seeds. We therefore identify grades $0$ to $3$ as the best-performing configuration by the
mean metrics, while the evidence for a separate benefit from each higher grade remains
limited. The progression is not a monotonic improvement as geometric features are added.

Extending the representation to the grade-4 magnitude feature reinforces this distinction
between the best mean result and a consistent performance gain. For the ancilla model at
$N{=}10$, the rejection at $50\%$ signal efficiency decreases from $35.7$ to $33.5$ when
this feature is added. At $N{=}8$, it changes only slightly in the opposite direction,
from $30.0$ to $30.2$ (Table~\ref{tab:results}). These results do not demonstrate a consistent
benefit from the tested four-volume-magnitude construction. This conclusion concerns that
specific feature, its encoding and its readout. Other ways of using grade-4 geometry remain
to be explored. In particular, the comparison does not establish that all useful geometric
information has been exhausted at grade~3.

Having varied the feature representation at fixed constituent count, we next ask what
happens when more of the jet is retained. Increasing $N$ from $8$ to $10$ supplies two
additional constituent momenta and adds two data qubits (Fig.~\ref{fig:constituent_scaling}).
For the ancilla model with grades $0$ to $3$, the AUC rises from $0.9353$ to $0.9439$ and
the rejection at $50\%$ signal efficiency from $30.0$ to $35.7$. Both quantities improve
across all four ancilla and grade configurations in Table~\ref{tab:results}. At $30\%$
signal efficiency, the changes are less uniform. For example, the ancilla model with grades
$0$ to $4$ changes from $90.4$ to $89.6$, well within its reported seed variation.

These comparisons favour the larger constituent set for AUC and rejection at the $50\%$
working point within the range studied. They also clarify what is being varied, since a larger
constituent set changes both the available kinematic information and the size of the
trainable circuit. The observed gains therefore describe the combined effect of these
changes. Establishing how performance scales beyond $N=10$, and separating the role of
additional constituents from that of circuit capacity, requires further comparisons.

\begin{figure}[H]
\centering
\includegraphics[width=0.95\linewidth]{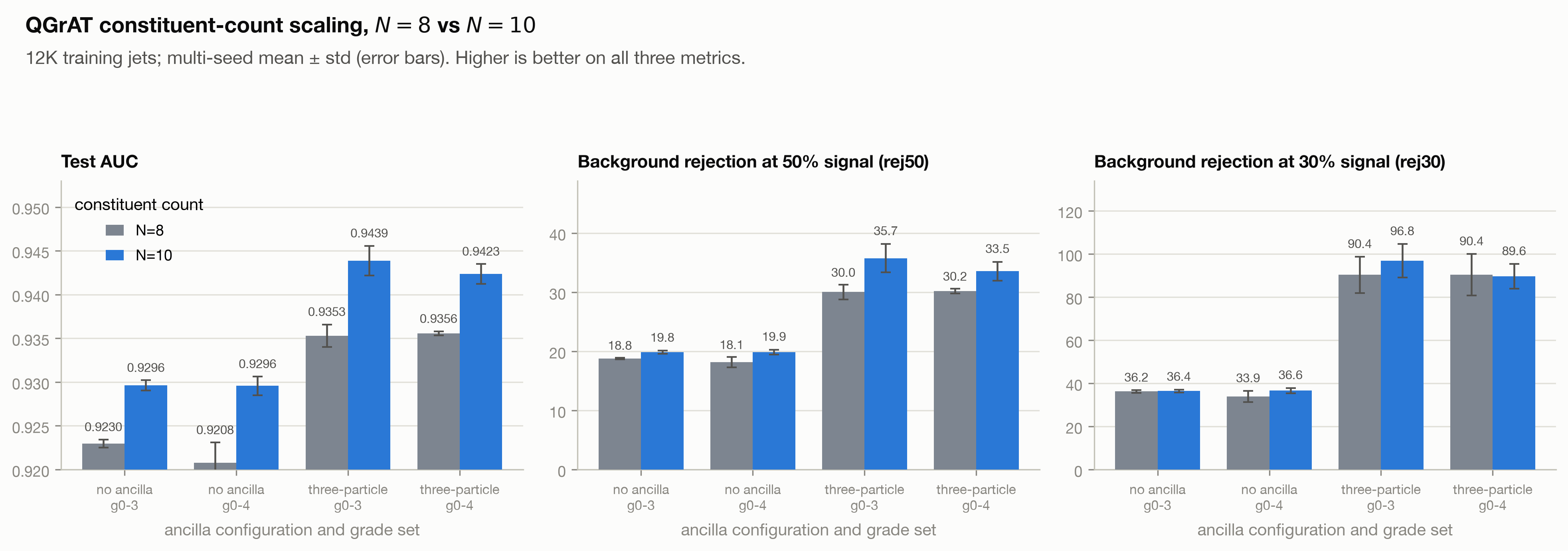}
\caption{Constituent-count scaling from $N=8$ to $N=10$ at 12K training jets, 100k test jets. Panels show test
AUC and background rejection at signal efficiencies $0.50$ and $0.30$, as multi-seed mean with
error bars for the standard deviation, for every ancilla configuration (no ancilla and
three-particle-product) at grades
$0$-$3$ and $0$-$4$. Higher is better on all three metrics.}
\label{fig:constituent_scaling}
\end{figure}

\section{Discussion and Conclusion}
\label{sec:conclusion}
\label{sec:discussion}

\begingroup
\clubpenalties 1 10000\relax
\widowpenalties 1 10000\relax

We have introduced QGrAT (Quantum Grade-Aware Transformer) to investigate how the
geometry of jet constituents can guide the design of a compact quantum classifier.
Its attention mechanism organises multivector features by Clifford grade and uses their
inner products to set the angles of pairwise entangling gates. A dedicated ancilla then
couples to triples of constituents through products of their pairwise Minkowski inner
products. In the Hadamard-test construction, these coefficients determine a coherent
phase response of the joint ancilla-data state, yet the ancilla does not directly measure the
individual products. After trainable mixing, an additional readout qubit gathers the
classification information into one Pauli expectation value, which is converted to a logit
by a learned scale and offset.

This construction brings spacetime geometric algebra into quantum attention, building on
classical geometric algebra networks, symmetry-informed quantum jet models and quantum
approaches to attention and Clifford algebra~\cite{lgatr,suda_neto_lie_eqgnn,comajoan_quantum_vision_jets,quantum_attention_transformer,trindade_clifford_qnn}.
To our knowledge, QGrAT is the first quantum jet tagger whose attention is organised
explicitly by the grades of $\Cl(1,3)$, with pairwise entangling-gate angles determined by
the corresponding grade-specific Clifford inner products. We combine this mechanism with
a coherent three-particle-product ancilla and a trainable in-circuit readout, establish the
symmetry properties of the couplings, and evaluate the contribution of the principal
architectural choices. Together, these developments provide an integrated framework for
translating geometric relations among jet constituents into quantum-circuit interactions
and testing their usefulness for classification.

Within the experimental setting of this paper, the strongest configuration with an
in-circuit readout uses grades $0$ to $3$ and the three-particle-product ancilla. At $N{=}10$
constituents, corresponding to $12$ qubits including the ancilla and readout, it reaches
AUC $0.9439{\pm}0.0017$ with $144$ nominal trainable parameters and $12{,}000$ training jets.
Its background rejection is $35.7{\pm}2.4$ at $50\%$ signal efficiency and $96.8{\pm}7.7$
at $30\%$ (Table~\ref{tab:results}). This is competitive with \lgatr{}-Slim at the same
constituent and parameter budgets. QGrAT has higher mean metrics and smaller reported
seed variation, although the differences in means lie within the larger seed spread of
the classical baseline. The larger \lgatr{} reference, with $1{,}215$ parameters, remains
clearly ahead at AUC $0.9579{\pm}0.0012$. These comparisons demonstrate useful tagging
performance with a small parameter budget, but matching parameter counts alone does not
establish equal model capacity or computational cost.

To understand where this performance comes from, Section~\ref{sec:systematic_experiments}
examines the model components separately. Among the tested changes, the three-particle-product
ancilla gives the clearest gain in background rejection, raising it from $19.9$ to $35.7$
at $50\%$ signal efficiency for grades $0$ to $3$ and $N{=}10$. The grade-1 control, which
retains the extra qubit but switches off its three-particle coupling, shows that the
explicit coupling is essential to the performance of that configuration. This supports
using a correlation motivated by the three-prong top decay to guide the circuit. It does
not, by itself, isolate the benefit of Lorentz symmetry from the additional interaction
or establish which physical information the trained model uses.

Adding geometric features gives a more nuanced picture. Without an ancilla, much of the
performance is already captured by the scalar feature $p_i\cdot P$. With the ancilla,
grades $0$ to $3$ give the highest mean metrics at $N{=}10$, but their gain over grades
$0$ to $1$ is small compared with the seed variation. The tested grade-4 magnitude feature
gives no consistent improvement, so the results do not support a monotonic gain from
adding higher grades. This finding concerns the chosen features and their encoding, but
it does not imply that four-volume information is intrinsically uninformative.
Increasing the constituent count from $8$ to $10$ improves AUC and rejection at $50\%$
signal efficiency across the compared configurations, while the changes at $30\%$ depend
on the grade and ancilla choices. Since this increase supplies both more input momenta
and a larger trainable circuit, the comparison measures their combined effect.

We also find that the useful classification signal can be extracted through either
readout strategy. The classical Hamiltonian-observable head gives higher mean AUC and
rejection for the ancilla model at grades $0$ to $3$, and comparable performance within
seed variation at grades $0$ to $4$ (Table~\ref{tab:hobs}). For grades $0$ to $3$, that
classical-head model uses $284$ parameters compared with $144$ for the in-circuit readout.
The latter therefore offers a more compact way to collect the signal, with readout
couplings growing linearly with the register size and a single measurement setting on
one qubit. The comparison does not establish a performance advantage for quantum readout,
and repeated circuit executions would still be needed to estimate its expectation value
on hardware.

The symmetry analysis clarifies what physical structure is built into this model.
For fixed constituent identities and co-transforming reference vectors, the grade-$0$
through grade-$3$ attention pairings and the triple-product coefficients are Lorentz
invariant. The grade-$4$ magnitude coupling additionally requires retaining the identity
of the selected leading constituent. These statements concern the coupling coefficients.
Encoding individual multivector components as rotation angles introduces frame dependence,
and selecting constituents again after a general Lorentz transformation can change the
inputs themselves (Section~\ref{sec:ising_invariance}). QGrAT thus uses Lorentz-invariant
relations to constrain its interactions while the complete classifier remains frame
dependent. Establishing how this design affects robustness under changes of frame would
require a separate test.

Training succeeds for the circuit sizes studied here, with non-vanishing gradient norms
and a pooled gradient-element variance above the plotted $2^{-N}$ reference
(Fig.~\ref{fig:scaling_diagnostics}). These observations describe the actual optimisation
runs. As specified in Section~\ref{sec:diagnostics}, the variance pools parameter
coordinates and training batches after gradient clipping, and it is not the variance of a
fixed derivative over random circuit initialisations used in standard barren-plateau
analyses~\cite{mcclean_barren}. The diagnostics therefore establish neither favourable
asymptotic scaling nor a well-conditioned optimisation landscape. A small trainable
parameter set and a single variational layer are plausible contributors to the observed
trainability, but the present comparisons do not isolate their effects. Analytical
guarantees for suitable permutation-equivariant circuits~\cite{schatzki_sn} do not
automatically apply to this architecture.

Several features of the evaluation delimit the conclusions. We study $8$ or $10$
constituents per jet from a dataset that stores up to $200$, and use one fixed training-set
size. This supports a comparison at restricted input and parameter budgets, but leaves
performance with larger constituent sets and the dependence on training-sample size open.
In addition, the triple-product coefficients are standardised within each batch, including
at evaluation time (Section~\ref{sec:phase_ancilla}). A jet's prediction consequently
depends on the other jets in its batch. The reported metrics apply to that protocol,
a classifier that assigns each jet a score independently would require a revised
normalisation, such as fixed training-set statistics, followed by retraining and evaluation.

All results are obtained in exact state-vector simulation. The small parameter count
should therefore be distinguished from the resources needed to execute the circuit:
pairwise attention requires $\binom{N}{2}$ entangling gates per active grade, and the
three-particle ancilla adds $\binom{N}{3}$ controlled $ZZZ$ rotations. Their decomposition
into native gates, limited connectivity, measurement statistics and hardware noise will
matter increasingly at larger $N$. No computational quantum advantage or hardware
performance gain is established by the present study. Larger-register simulations and
hardware measurements are needed to determine whether the observed discrimination can
be retained at a practical execution cost.

These findings suggest concrete directions for further development. Alternative
higher-grade features can test whether the circuit can make better use of the geometry
already present in the retained momenta. Comparing ZZ attention with XX or YY couplings
driven by the same invariant angles would test the role of the entangling operator while
preserving the symmetry of its coefficients. The present results establish no preference
among these choices. Comparisons with classical models supplied with the same geometric
and triple-product features would further clarify the contribution of the circuit to the
observed gains~\cite{schreiber_classical_surrogates}. QGrAT provides a concrete starting point for these studies: a compact
quantum tagger with explicitly defined geometric interactions, competitive performance
against the parameter-matched classical reference, and direct numerical evidence that
an explicit three-particle coupling can substantially improve background rejection.

\par
\endgroup

\subsection*{Acknowledgements}

We are grateful to Tilman Plehn for suggesting this project and for his generous
help, insightful advice, and encouragement. His support and thoughtful discussions have
been instrumental to the development of this work. \\
The authors acknowledge support by the state of Baden-Württemberg through bwHPC and the German Research Foundation (DFG) through grant no INST 39/963-1 FUGG (bwForCluster NEMO).

\clearpage

\bibliographystyle{unsrt}
\bibliography{refs}

@article{kasieczka_toptag,
  author        = {Kasieczka, Gregor and Plehn, Tilman and Butter, Anja and others},
  title         = {The Machine Learning landscape of top taggers},
  journal       = {SciPost Phys.},
  volume        = {7},
  pages         = {014},
  year          = {2019},
  doi           = {10.21468/SciPostPhys.7.1.014},
  eprint        = {1902.09914},
  archivePrefix = {arXiv},
  primaryClass  = {hep-ph},
  note          = {\href{https://arxiv.org/abs/1902.09914}{arXiv:1902.09914 [hep-ph]}},
}

@misc{toptag_dataset,
  author        = {Kasieczka, Gregor and Plehn, Tilman and Thompson, Jennifer and Russel, Michael},
  title         = {Top Quark Tagging Reference Dataset},
  howpublished  = {Zenodo},
  year          = {2019},
  doi           = {10.5281/zenodo.2603256},
  note          = {\href{https://doi.org/10.5281/zenodo.2603256}{doi:10.5281/zenodo.2603256}},
}

@article{lorentznet,
  author        = {Gong, Shiqi and Meng, Qi and Zhang, Jue and Qu, Huilin and Li, Congqiao
                   and Qian, Sitian and Du, Weitao and Ma, Zhi-Ming and Liu, Tie-Yan},
  title         = {An efficient {Lorentz} equivariant graph neural network for jet tagging},
  journal       = {JHEP},
  volume        = {07},
  pages         = {030},
  year          = {2022},
  doi           = {10.1007/JHEP07(2022)030},
  eprint        = {2201.08187},
  archivePrefix = {arXiv},
  primaryClass  = {hep-ph},
  note          = {\href{https://arxiv.org/abs/2201.08187}{arXiv:2201.08187 [hep-ph]}},
}

@misc{pelican,
  author        = {Bogatskiy, Alexander and Hoffman, Timothy and Miller, David W.
                   and Offermann, Jan T.},
  title         = {{PELICAN}: Permutation Equivariant and {Lorentz} Invariant or Covariant
                   Aggregator Network for Particle Physics},
  howpublished  = {\href{https://arxiv.org/abs/2211.00454}{arXiv:2211.00454 [hep-ph]}},
  year          = {2022},
  eprint        = {2211.00454},
  archivePrefix = {arXiv},
  primaryClass  = {hep-ph},
}

@inproceedings{particle_transformer,
  author        = {Qu, Huilin and Li, Congqiao and Qian, Sitian},
  title         = {Particle Transformer for Jet Tagging},
  booktitle     = {Proceedings of the 39th International Conference on Machine Learning},
  series        = {Proceedings of Machine Learning Research},
  volume        = {162},
  pages         = {18281--18292},
  year          = {2022},
  eprint        = {2202.03772},
  archivePrefix = {arXiv},
  primaryClass  = {hep-ph},
  note          = {\href{https://arxiv.org/abs/2202.03772}{arXiv:2202.03772 [hep-ph]}},
}

@inproceedings{gatr,
  author        = {Brehmer, Johann and de Haan, Pim and Behrends, S{\"o}nke and Cohen, Taco},
  title         = {Geometric Algebra Transformer},
  booktitle     = {Advances in Neural Information Processing Systems},
  volume        = {36},
  pages         = {35472--35496},
  year          = {2023},
  doi           = {10.52202/075280-1541},
  eprint        = {2305.18415},
  archivePrefix = {arXiv},
  primaryClass  = {cs.LG},
  note          = {\href{https://arxiv.org/abs/2305.18415}{arXiv:2305.18415 [cs.LG]}},
}

@inproceedings{lgatr,
  author        = {Spinner, Jonas and Bres{\'o}, V{\'\i}ctor and de Haan, Pim and Plehn, Tilman
                   and Thaler, Jesse and Brehmer, Johann},
  title         = {{Lorentz}-Equivariant Geometric Algebra Transformers for High-Energy Physics},
  booktitle     = {Advances in Neural Information Processing Systems},
  volume        = {37},
  pages         = {22178--22205},
  year          = {2024},
  doi           = {10.52202/079017-0699},
  eprint        = {2405.14806},
  archivePrefix = {arXiv},
  primaryClass  = {physics.data-an},
  note          = {\href{https://arxiv.org/abs/2405.14806}{arXiv:2405.14806 [physics.data-an]}},
}

@article{larocca_gqml,
  author        = {Larocca, Mart{\'\i}n and Sauvage, Fr{\'e}d{\'e}ric and Sbahi, Faris M.
                   and Verdon, Guillaume and Coles, Patrick J. and Cerezo, M.},
  title         = {Group-Invariant Quantum Machine Learning},
  journal       = {PRX Quantum},
  volume        = {3},
  number        = {3},
  pages         = {030341},
  year          = {2022},
  doi           = {10.1103/PRXQuantum.3.030341},
  eprint        = {2205.02261},
  archivePrefix = {arXiv},
  primaryClass  = {quant-ph},
  note          = {\href{https://arxiv.org/abs/2205.02261}{arXiv:2205.02261 [quant-ph]}},
}

@article{meyer_group_invariant,
  author        = {Meyer, Johannes Jakob and Mularski, Marian and Gil-Fuster, Elies
                   and Mele, Antonio Anna and Arzani, Francesco and Wilms, Alissa and Eisert, Jens},
  title         = {Exploiting Symmetry in Variational Quantum Machine Learning},
  journal       = {PRX Quantum},
  volume        = {4},
  number        = {1},
  pages         = {010328},
  year          = {2023},
  doi           = {10.1103/PRXQuantum.4.010328},
  eprint        = {2205.06217},
  archivePrefix = {arXiv},
  primaryClass  = {quant-ph},
  note          = {\href{https://arxiv.org/abs/2205.06217}{arXiv:2205.06217 [quant-ph]}},
}

@article{nguyen_equivariant,
  author        = {Nguyen, Quynh T. and Schatzki, Louis and Braccia, Paolo and
                   Ragone, Michael and Coles, Patrick J. and Sauvage, Fr{\'e}d{\'e}ric and
                   Larocca, Mart{\'\i}n and Cerezo, M.},
  title         = {Theory for Equivariant Quantum Neural Networks},
  journal       = {PRX Quantum},
  volume        = {5},
  number        = {2},
  pages         = {020328},
  year          = {2024},
  doi           = {10.1103/PRXQuantum.5.020328},
  eprint        = {2210.08566},
  archivePrefix = {arXiv},
  primaryClass  = {quant-ph},
  note          = {\href{https://arxiv.org/abs/2210.08566}{arXiv:2210.08566 [quant-ph]}},
}

@article{schatzki_sn,
  author        = {Schatzki, Louis and Larocca, Mart{\'\i}n and Nguyen, Quynh T.
                   and Sauvage, Fr{\'e}d{\'e}ric and Cerezo, M.},
  title         = {Theoretical guarantees for permutation-equivariant quantum neural networks},
  journal       = {npj Quantum Inf.},
  volume        = {10},
  pages         = {12},
  year          = {2024},
  doi           = {10.1038/s41534-024-00804-1},
  eprint        = {2210.09974},
  archivePrefix = {arXiv},
  primaryClass  = {quant-ph},
  note          = {\href{https://arxiv.org/abs/2210.09974}{arXiv:2210.09974 [quant-ph]}},
}

@article{blance_spannowsky_vqc,
  author        = {Blance, Andrew and Spannowsky, Michael},
  title         = {Quantum Machine Learning for Particle Physics using a Variational Quantum Classifier},
  journal       = {JHEP},
  volume        = {02},
  pages         = {212},
  year          = {2021},
  doi           = {10.1007/JHEP02(2021)212},
  eprint        = {2010.07335},
  archivePrefix = {arXiv},
  primaryClass  = {hep-ph},
  note          = {\href{https://arxiv.org/abs/2010.07335}{arXiv:2010.07335 [hep-ph]}},
}

@article{ngairangbam_quantum_autoencoder,
  author        = {Ngairangbam, Vishal S. and Spannowsky, Michael and Takeuchi, Michihisa},
  title         = {Anomaly detection in high-energy physics using a quantum autoencoder},
  journal       = {Phys. Rev. D},
  volume        = {105},
  number        = {9},
  pages         = {095004},
  year          = {2022},
  doi           = {10.1103/PhysRevD.105.095004},
  eprint        = {2112.04958},
  archivePrefix = {arXiv},
  primaryClass  = {hep-ph},
  note          = {\href{https://arxiv.org/abs/2112.04958}{arXiv:2112.04958 [hep-ph]}},
}

@misc{ngairangbam_geometric_qnn,
  author        = {Ngairangbam, Vishal S. and Spannowsky, Michael},
  title         = {From Reachability to Learnability: Geometric Design Principles for Quantum Neural Networks},
  howpublished  = {\href{https://arxiv.org/abs/2603.03071}{arXiv:2603.03071 [quant-ph]}},
  year          = {2026},
  reportNumber  = {IPPP/26/22},
  eprint        = {2603.03071},
  archivePrefix = {arXiv},
  primaryClass  = {quant-ph},
  doi           = {10.48550/arXiv.2603.03071},
}

@misc{ngairangbam_generic_reachability,
  author        = {Ngairangbam, Vishal S. and Spannowsky, Michael},
  title         = {Equivalence of maximal and generic reachability for non-universal Variational Quantum Circuits},
  howpublished  = {\href{https://arxiv.org/abs/2609.27053}{arXiv:2609.27053 [quant-ph]}},
  year          = {2026},
  eprint        = {2609.27053},
  archivePrefix = {arXiv},
  primaryClass  = {quant-ph},
  doi           = {10.48550/arXiv.2609.27053},
}

@article{abel_completely_quantum_nn,
  author        = {Abel, Steve and Criado, Juan C. and Spannowsky, Michael},
  title         = {Completely Quantum Neural Networks},
  journal       = {Phys. Rev. A},
  volume        = {106},
  number        = {2},
  pages         = {022601},
  year          = {2022},
  doi           = {10.1103/PhysRevA.106.022601},
  eprint        = {2202.11727},
  archivePrefix = {arXiv},
  primaryClass  = {quant-ph},
  note          = {\href{https://arxiv.org/abs/2202.11727}{arXiv:2202.11727 [quant-ph]}},
}

@article{bal_one_particle_one_qubit,
  author        = {Bal, Aritra and Klute, Markus and Maier, Benedikt and Oughton, Melik
                   and Pezone, Eric and Spannowsky, Michael},
  title         = {One particle -- one qubit: Particle physics data encoding for quantum machine learning},
  journal       = {Phys. Rev. D},
  volume        = {112},
  number        = {7},
  pages         = {076004},
  year          = {2025},
  doi           = {10.1103/l8y2-87vq},
  eprint        = {2502.17301},
  archivePrefix = {arXiv},
  primaryClass  = {hep-ph},
  note          = {\href{https://arxiv.org/abs/2502.17301}{arXiv:2502.17301 [hep-ph]}},
}

@article{belis_guided_quantum_compression,
  author        = {Belis, Vasilis and Odagiu, Patrick and Grossi, Michele and Reiter, Florentin
                   and Dissertori, G{\"u}nther and Vallecorsa, Sofia},
  title         = {Guided Quantum Compression for High Dimensional Data Classification},
  journal       = {Mach. Learn.: Sci. Technol.},
  volume        = {5},
  number        = {3},
  pages         = {035010},
  year          = {2024},
  doi           = {10.1088/2632-2153/ad5fdd},
  eprint        = {2402.09524},
  archivePrefix = {arXiv},
  primaryClass  = {quant-ph},
  note          = {\href{https://arxiv.org/abs/2402.09524}{arXiv:2402.09524 [quant-ph]}},
}

@inproceedings{bermot_quantum_gan_hep,
  author        = {Bermot, Elie and Zoufal, Christa and Grossi, Michele and Schuhmacher, Julian
                   and Tacchino, Francesco and Vallecorsa, Sofia and Tavernelli, Ivano},
  title         = {Quantum Generative Adversarial Networks for Anomaly Detection in High Energy Physics},
  booktitle     = {2023 IEEE International Conference on Quantum Computing and Engineering (QCE)},
  publisher     = {IEEE},
  volume        = {1},
  pages         = {331--341},
  year          = {2023},
  doi           = {10.1109/QCE57702.2023.00045},
  eprint        = {2304.14439},
  archivePrefix = {arXiv},
  primaryClass  = {quant-ph},
  note          = {\href{https://arxiv.org/abs/2304.14439}{arXiv:2304.14439 [quant-ph]}},
}

@article{funcke_dimensional_expressivity,
  author        = {Funcke, Lena and Hartung, Tobias and Jansen, Karl and K{\"u}hn, Stefan
                   and Stornati, Paolo},
  title         = {Dimensional Expressivity Analysis of Parametric Quantum Circuits},
  journal       = {Quantum},
  volume        = {5},
  pages         = {422},
  year          = {2021},
  doi           = {10.22331/q-2021-03-29-422},
  eprint        = {2011.03532},
  archivePrefix = {arXiv},
  primaryClass  = {quant-ph},
  note          = {\href{https://arxiv.org/abs/2011.03532}{arXiv:2011.03532 [quant-ph]}},
}

@article{forestano_invariant_equivariant_qgnn,
  author        = {Forestano, Roy T. and Comajoan Cara, Mar{\c c}al and Dahale, Gopal Ramesh
                   and Dong, Zhongtian and Gleyzer, Sergei and Justice, Daniel
                   and Kong, Kyoungchul and Magorsch, Tom and Matchev, Konstantin T.
                   and Matcheva, Katia and Unlu, Eyup B.},
  title         = {A Comparison between Invariant and Equivariant Classical and Quantum Graph Neural Networks},
  journal       = {Axioms},
  volume        = {13},
  number        = {3},
  pages         = {160},
  year          = {2024},
  doi           = {10.3390/axioms13030160},
  eprint        = {2311.18672},
  archivePrefix = {arXiv},
  primaryClass  = {quant-ph},
  note          = {\href{https://arxiv.org/abs/2311.18672}{arXiv:2311.18672 [quant-ph]}},
}

@article{comajoan_quantum_vision_jets,
  author        = {Comajoan Cara, Mar{\c c}al and Dahale, Gopal Ramesh and Dong, Zhongtian
                   and Forestano, Roy T. and Gleyzer, Sergei and Justice, Daniel
                   and Kong, Kyoungchul and Magorsch, Tom and Matchev, Konstantin T.
                   and Matcheva, Katia and Unlu, Eyup B.},
  title         = {Quantum Vision Transformers for Quark--Gluon Classification},
  journal       = {Axioms},
  volume        = {13},
  number        = {5},
  pages         = {323},
  year          = {2024},
  doi           = {10.3390/axioms13050323},
  eprint        = {2405.10284},
  archivePrefix = {arXiv},
  primaryClass  = {quant-ph},
  note          = {\href{https://arxiv.org/abs/2405.10284}{arXiv:2405.10284 [quant-ph]}},
}

@article{wu_quantum_hep,
  author        = {Wu, Sau Lan and Sun, Shaojun and Guan, Wen and others},
  title         = {Application of quantum machine learning using the quantum kernel algorithm
                   on high energy physics analysis at the {LHC}},
  journal       = {Phys. Rev. Res.},
  volume        = {3},
  number        = {3},
  pages         = {033221},
  year          = {2021},
  doi           = {10.1103/PhysRevResearch.3.033221},
  eprint        = {2104.05059},
  archivePrefix = {arXiv},
  primaryClass  = {quant-ph},
  note          = {\href{https://arxiv.org/abs/2104.05059}{arXiv:2104.05059 [quant-ph]}},
}

@inproceedings{suda_neto_lie_eqgnn,
  author        = {Suda Neto, Jogi and Forestano, Roy T. and Gleyzer, Sergei
                   and Kong, Kyoungchul and Matchev, Konstantin T. and Matcheva, Katia},
  title         = {{Lie}-Equivariant Quantum Graph Neural Networks},
  booktitle     = {Machine Learning with New Compute Paradigms (MLNCP) Workshop at NeurIPS 2024},
  year          = {2024},
  eprint        = {2411.15315},
  archivePrefix = {arXiv},
  primaryClass  = {quant-ph},
  note          = {\href{https://arxiv.org/abs/2411.15315}{arXiv:2411.15315 [quant-ph]}},
}

@article{li_enforcing_symmetries,
  author        = {Li, Zhelun and Nagano, Lento and Terashi, Koji},
  title         = {Enforcing exact permutation and rotational symmetries in the application of
                   quantum neural networks on point cloud datasets},
  journal       = {Phys. Rev. Res.},
  volume        = {6},
  number        = {4},
  pages         = {043028},
  year          = {2024},
  doi           = {10.1103/PhysRevResearch.6.043028},
  eprint        = {2405.11150},
  archivePrefix = {arXiv},
  primaryClass  = {quant-ph},
  note          = {\href{https://arxiv.org/abs/2405.11150}{arXiv:2405.11150 [quant-ph]}},
}

@misc{elhag_qcnn_jets,
  author        = {Elhag, Hala and Hartung, Tobias and Jansen, Karl and Nagano, Lento
                   and Pirina, Giorgio Menicagli and Di Tucci, Alice},
  title         = {Quantum convolutional neural networks for jet images classification},
  howpublished  = {\href{https://arxiv.org/abs/2408.08701}{arXiv:2408.08701 [quant-ph]}},
  year          = {2024},
  eprint        = {2408.08701},
  archivePrefix = {arXiv},
  primaryClass  = {quant-ph},
}

@article{lp2b_encoding,
  author        = {Napolitano, Fabrizio and Della Penna, Luca and Tedeschi, Tommaso and Fan{\`o}, Livio},
  title         = {{Lund} Plane to {Bloch} ({LP$^{2}$B}) encoding for object and polarization tagging
                   with quantum jet substructure},
  journal       = {Eur. Phys. J. C},
  volume        = {86},
  number        = {7},
  pages         = {909},
  year          = {2026},
  doi           = {10.1140/epjc/s10052-026-16142-9},
  eprint        = {2604.18613},
  archivePrefix = {arXiv},
  primaryClass  = {quant-ph},
  note          = {\href{https://arxiv.org/abs/2604.18613}{arXiv:2604.18613 [quant-ph]}},
}

@article{trindade_clifford_qnn,
  author        = {Trindade, Marco A. S. and Lula-Rocha, Vin{\'\i}cius N. A. and Floquet, S.},
  title         = {{Clifford} Algebras, Quantum Neural Networks and Generalized Quantum {Fourier} Transform},
  journal       = {Adv. Appl. Clifford Algebras},
  volume        = {33},
  number        = {3},
  pages         = {38},
  year          = {2023},
  doi           = {10.1007/s00006-023-01279-7},
  eprint        = {2206.01808},
  archivePrefix = {arXiv},
  primaryClass  = {quant-ph},
  note          = {\href{https://arxiv.org/abs/2206.01808}{arXiv:2206.01808 [quant-ph]}},
}

@book{doran_lasenby,
  author        = {Doran, Chris and Lasenby, Anthony},
  title         = {Geometric Algebra for Physicists},
  publisher     = {Cambridge University Press},
  year          = {2003},
  doi           = {10.1017/CBO9780511807497},
}

@article{quantum_attention_transformer,
  author        = {Li, Guangxi and Zhao, Xuanqiang and Wang, Xin},
  title         = {Quantum self-attention neural networks for text classification},
  journal       = {Sci. China Inf. Sci.},
  volume        = {67},
  number        = {4},
  pages         = {142501},
  year          = {2024},
  doi           = {10.1007/s11432-023-3879-7},
  eprint        = {2205.05625},
  archivePrefix = {arXiv},
  primaryClass  = {quant-ph},
  note          = {\href{https://arxiv.org/abs/2205.05625}{arXiv:2205.05625 [quant-ph]}},
}

@article{mcclean_barren,
  author        = {McClean, Jarrod R. and Boixo, Sergio and Smelyanskiy, Vadim N.
                   and Babbush, Ryan and Neven, Hartmut},
  title         = {Barren plateaus in quantum neural network training landscapes},
  journal       = {Nat. Commun.},
  volume        = {9},
  pages         = {4812},
  year          = {2018},
  doi           = {10.1038/s41467-018-07090-4},
  eprint        = {1803.11173},
  archivePrefix = {arXiv},
  primaryClass  = {quant-ph},
  note          = {\href{https://arxiv.org/abs/1803.11173}{arXiv:1803.11173 [quant-ph]}},
}

@article{perez_garcia_reuploading,
  author        = {P{\'e}rez-Salinas, Adri{\'a}n and Cervera-Lierta, Alba and Gil-Fuster, Elies
                   and Latorre, Jos{\'e} I.},
  title         = {Data re-uploading for a universal quantum classifier},
  journal       = {Quantum},
  volume        = {4},
  pages         = {226},
  year          = {2020},
  doi           = {10.22331/q-2020-02-06-226},
  eprint        = {1907.02085},
  archivePrefix = {arXiv},
  primaryClass  = {quant-ph},
  note          = {\href{https://arxiv.org/abs/1907.02085}{arXiv:1907.02085 [quant-ph]}},
}

@article{cerezo_vqe_review,
  author        = {Cerezo, M. and Arrasmith, Andrew and Babbush, Ryan
                   and Benjamin, Simon C. and Endo, Suguru and Fujii, Keisuke
                   and McClean, Jarrod R. and Mitarai, Kosuke and Yuan, Xiao
                   and Cincio, Lukasz and Coles, Patrick J.},
  title         = {Variational Quantum Algorithms},
  journal       = {Nat. Rev. Phys.},
  volume        = {3},
  number        = {9},
  pages         = {625--644},
  year          = {2021},
  doi           = {10.1038/s42254-021-00348-9},
  eprint        = {2012.09265},
  archivePrefix = {arXiv},
  primaryClass  = {quant-ph},
  note          = {\href{https://arxiv.org/abs/2012.09265}{arXiv:2012.09265 [quant-ph]}},
}

@book{hestenes_sta,
  author        = {Hestenes, David},
  title         = {Space-Time Algebra},
  publisher     = {Gordon and Breach},
  address       = {New York},
  year          = {1966},
  note          = {2nd edition: Birkh{\"a}user, Cham, 2015},
}

@book{lounesto_clifford,
  author        = {Lounesto, Pertti},
  title         = {Clifford Algebras and Spinors},
  series        = {London Mathematical Society Lecture Note Series},
  volume        = {286},
  publisher     = {Cambridge University Press},
  edition       = {2nd},
  year          = {2001},
  doi           = {10.1017/CBO9780511526022},
}

@article{cleve1998quantum,
  author        = {Cleve, Richard and Ekert, Artur and Macchiavello, Chiara and Mosca, Michele},
  title         = {Quantum algorithms revisited},
  journal       = {Proc. R. Soc. Lond. A},
  volume        = {454},
  number        = {1969},
  pages         = {339--354},
  year          = {1998},
  doi           = {10.1098/rspa.1998.0164},
  eprint        = {quant-ph/9708016},
  archivePrefix = {arXiv},
  note          = {\href{https://arxiv.org/abs/quant-ph/9708016}{arXiv:quant-ph/9708016}},
}

@book{nielsen2010quantum,
  author        = {Nielsen, Michael A. and Chuang, Isaac L.},
  title         = {Quantum Computation and Quantum Information},
  publisher     = {Cambridge University Press},
  edition       = {10th anniversary},
  year          = {2010},
  doi           = {10.1017/CBO9780511976667},
}

@article{ekert2002direct,
  author        = {Ekert, Artur K. and Alves, Carolina Moura and Oi, Daniel K. L.
                   and Horodecki, Micha{\l} and Horodecki, Pawe{\l} and Kwek, L. C.},
  title         = {Direct estimations of linear and nonlinear functionals of a quantum state},
  journal       = {Phys. Rev. Lett.},
  volume        = {88},
  number        = {21},
  pages         = {217901},
  year          = {2002},
  doi           = {10.1103/PhysRevLett.88.217901},
  eprint        = {quant-ph/0203016},
  archivePrefix = {arXiv},
  note          = {\href{https://arxiv.org/abs/quant-ph/0203016}{arXiv:quant-ph/0203016}},
}

@article{schuld_circuit_centric,
  author        = {Schuld, Maria and Bocharov, Alex and Svore, Krysta M. and Wiebe, Nathan},
  title         = {Circuit-centric quantum classifiers},
  journal       = {Phys. Rev. A},
  volume        = {101},
  number        = {3},
  pages         = {032308},
  year          = {2020},
  doi           = {10.1103/PhysRevA.101.032308},
  eprint        = {1804.00633},
  archivePrefix = {arXiv},
  primaryClass  = {quant-ph},
  note          = {\href{https://arxiv.org/abs/1804.00633}{arXiv:1804.00633 [quant-ph]}},
}

@article{havlicek_feature_map,
  author        = {Havl{\'\i}{\v{c}}ek, Vojt{\v{e}}ch and C{\'o}rcoles, Antonio D. and Temme, Kristan
                   and Harrow, Aram W. and Kandala, Abhinav and Chow, Jerry M. and Gambetta, Jay M.},
  title         = {Supervised learning with quantum-enhanced feature spaces},
  journal       = {Nature},
  volume        = {567},
  number        = {7747},
  pages         = {209--212},
  year          = {2019},
  doi           = {10.1038/s41586-019-0980-2},
  eprint        = {1804.11326},
  archivePrefix = {arXiv},
  primaryClass  = {quant-ph},
  note          = {\href{https://arxiv.org/abs/1804.11326}{arXiv:1804.11326 [quant-ph]}},
}

@article{sim_expressibility,
  author        = {Sim, Sukin and Johnson, Peter D. and Aspuru-Guzik, Al{\'a}n},
  title         = {Expressibility and entangling capability of parameterized quantum circuits
                   for hybrid quantum-classical algorithms},
  journal       = {Adv. Quantum Technol.},
  volume        = {2},
  number        = {12},
  pages         = {1900070},
  year          = {2019},
  doi           = {10.1002/qute.201900070},
  eprint        = {1905.10876},
  archivePrefix = {arXiv},
  primaryClass  = {quant-ph},
  note          = {\href{https://arxiv.org/abs/1905.10876}{arXiv:1905.10876 [quant-ph]}},
}

@article{qcgnn_jet,
  author        = {Chen, Yi-An and Chen, Kai-Feng},
  title         = {Jet discrimination with a quantum complete graph neural network},
  journal       = {Phys. Rev. D},
  volume        = {111},
  number        = {1},
  pages         = {016020},
  year          = {2025},
  doi           = {10.1103/PhysRevD.111.016020},
  eprint        = {2403.04990},
  archivePrefix = {arXiv},
  primaryClass  = {hep-ph},
  note          = {\href{https://arxiv.org/abs/2403.04990}{arXiv:2403.04990 [hep-ph]}},
}

@incollection{Spannowsky:2025tvb,
  author        = {Spannowsky, Michael},
  editor        = {Fischer, Christian S.},
  title         = {Searching for New Physics with the {Large Hadron Collider}},
  booktitle     = {Encyclopedia of Particle Physics},
  publisher     = {Elsevier},
  volume        = {1},
  pages         = {182--200},
  year          = {2026},
  doi           = {10.1016/B978-0-443-26598-3.00072-9},
  eprint        = {2510.03704},
  archivePrefix = {arXiv},
  primaryClass  = {hep-ph},
  reportNumber  = {IPPP/25/62},
  note          = {\href{https://arxiv.org/abs/2510.03704}{arXiv:2510.03704 [hep-ph]}},
}

@book{Marzani:2019hun,
  author        = {Marzani, Simone and Soyez, Gregory and Spannowsky, Michael},
  title         = {Looking inside jets: an introduction to jet substructure and boosted-object phenomenology},
  publisher     = {Springer},
  series        = {Lecture Notes in Physics},
  volume        = {1053},
  edition       = {2nd},
  year          = {2026},
  doi           = {10.1007/978-3-032-27190-7},
  eprint        = {1901.10342},
  archivePrefix = {arXiv},
  primaryClass  = {hep-ph},
  note          = {\href{https://arxiv.org/abs/1901.10342}{arXiv:1901.10342 [hep-ph]}},
}

@article{Kogler:2018hem,
  author        = {Kogler, Roman and Nachman, Benjamin and Schmidt, Alexander and others},
  title         = {Jet substructure at the {Large Hadron Collider}},
  journal       = {Rev. Mod. Phys.},
  volume        = {91},
  number        = {4},
  pages         = {045003},
  year          = {2019},
  doi           = {10.1103/RevModPhys.91.045003},
  eprint        = {1803.06991},
  archivePrefix = {arXiv},
  primaryClass  = {hep-ex},
  note          = {\href{https://arxiv.org/abs/1803.06991}{arXiv:1803.06991 [hep-ex]}},
}

@inproceedings{Bogatskiy:2020pex,
  author        = {Bogatskiy, Alexander and Anderson, Brandon and Offermann, Jan T.
                   and Roussi, Marwah and Miller, David W. and Kondor, Risi},
  title         = {{Lorentz} Group Equivariant Neural Network for Particle Physics},
  booktitle     = {Proceedings of the 37th International Conference on Machine Learning},
  series        = {Proceedings of Machine Learning Research},
  volume        = {119},
  pages         = {992--1002},
  year          = {2020},
  eprint        = {2006.04780},
  archivePrefix = {arXiv},
  primaryClass  = {hep-ph},
  note          = {\href{https://arxiv.org/abs/2006.04780}{arXiv:2006.04780 [hep-ph]}},
}

@article{Brehmer:2024LHC,
  author        = {Brehmer, Johann and Bres{\'o}, V{\'\i}ctor and de Haan, Pim and Plehn, Tilman
                   and Qu, Huilin and Spinner, Jonas and Thaler, Jesse},
  title         = {A {Lorentz}-equivariant transformer for all of the {LHC}},
  journal       = {SciPost Phys.},
  volume        = {19},
  number        = {4},
  pages         = {108},
  year          = {2025},
  doi           = {10.21468/SciPostPhys.19.4.108},
  eprint        = {2411.00446},
  archivePrefix = {arXiv},
  primaryClass  = {hep-ph},
  note          = {\href{https://arxiv.org/abs/2411.00446}{arXiv:2411.00446 [hep-ph]}},
}

@article{maier_photonic_qelm,
  author        = {Maier, Benedikt and Spannowsky, Michael and Williams, Simon},
  title         = {Continuous-variable photonic quantum extreme learning machines for fast collider-data selection},
  journal       = {EPJ Quantum Technol.},
  volume        = {13},
  number        = {1},
  pages         = {70},
  year          = {2026},
  doi           = {10.1140/epjqt/s40507-026-00507-w},
  eprint        = {2510.13994},
  archivePrefix = {arXiv},
  primaryClass  = {quant-ph},
  note          = {\href{https://arxiv.org/abs/2510.13994}{arXiv:2510.13994 [quant-ph]}},
}

@article{blance_photonic_graphs,
  author        = {Blance, Andrew and Spannowsky, Michael},
  title         = {Unsupervised Event Classification with Graphs on Classical and Photonic Quantum Computers},
  journal       = {JHEP},
  volume        = {08},
  pages         = {170},
  year          = {2021},
  doi           = {10.1007/JHEP08(2021)170},
  eprint        = {2103.03897},
  archivePrefix = {arXiv},
  primaryClass  = {hep-ph},
  note          = {\href{https://arxiv.org/abs/2103.03897}{arXiv:2103.03897 [hep-ph]}},
}

@article{incudini_quantum_kernels,
  author        = {Incudini, Massimiliano and Lizzio Bosco, Daniele and Martini, Francesco
                   and Grossi, Michele and Serra, Giuseppe and Di Pierro, Alessandra},
  title         = {Automatic and Effective Discovery of Quantum Kernels},
  journal       = {IEEE Trans. Emerg. Top. Comput. Intell.},
  volume        = {10},
  number        = {3},
  pages         = {2210--2219},
  year          = {2026},
  doi           = {10.1109/TETCI.2024.3499993},
  eprint        = {2209.11144},
  archivePrefix = {arXiv},
  primaryClass  = {quant-ph},
  note          = {\href{https://arxiv.org/abs/2209.11144}{arXiv:2209.11144 [quant-ph]}},
}

@article{belis_latent_quantum_anomaly,
  author        = {Belis, Vasilis and Wo{\'z}niak, Kinga Anna and Puljak, Ema
                   and Barkoutsos, Panagiotis and Dissertori, G{\"u}nther and Grossi, Michele
                   and Pierini, Maurizio and Reiter, Florentin and Tavernelli, Ivano
                   and Vallecorsa, Sofia},
  title         = {Quantum anomaly detection in the latent space of proton collision events at the {LHC}},
  journal       = {Commun. Phys.},
  volume        = {7},
  pages         = {334},
  year          = {2024},
  doi           = {10.1038/s42005-024-01811-6},
  eprint        = {2301.10780},
  archivePrefix = {arXiv},
  primaryClass  = {quant-ph},
  note          = {\href{https://arxiv.org/abs/2301.10780}{arXiv:2301.10780 [quant-ph]}},
}

@article{schreiber_classical_surrogates,
  author        = {Schreiber, Franz J. and Eisert, Jens and Meyer, Johannes Jakob},
  title         = {Classical Surrogates for Quantum Learning Models},
  journal       = {Phys. Rev. Lett.},
  volume        = {131},
  number        = {10},
  pages         = {100803},
  year          = {2023},
  doi           = {10.1103/PhysRevLett.131.100803},
  eprint        = {2206.11740},
  archivePrefix = {arXiv},
  primaryClass  = {quant-ph},
  note          = {\href{https://arxiv.org/abs/2206.11740}{arXiv:2206.11740 [quant-ph]}},
}

@article{felser_quantum_inspired_hep,
  author        = {Felser, Timo and Trenti, Marco and Sestini, Lorenzo and Gianelle, Alessio
                   and Zuliani, Davide and Lucchesi, Donatella and Montangero, Simone},
  title         = {Quantum-inspired machine learning on high-energy physics data},
  journal       = {npj Quantum Inf.},
  volume        = {7},
  pages         = {111},
  year          = {2021},
  doi           = {10.1038/s41534-021-00443-w},
  eprint        = {2004.13747},
  archivePrefix = {arXiv},
  primaryClass  = {stat.ML},
  note          = {\href{https://arxiv.org/abs/2004.13747}{arXiv:2004.13747 [stat.ML]}},
}
\end{document}